\documentclass[conference]{IEEEtran}
\IEEEoverridecommandlockouts
\usepackage{graphicx}
\usepackage[skip=1ex]{caption}
\usepackage[skip=1ex]{subcaption} 
\usepackage{amsthm}
\usepackage[table]{xcolor}
\usepackage{enumitem}
\usepackage[vlined,ruled,linesnumbered]{algorithm2e}
\usepackage{algpseudocode}  
\usepackage{amsmath,amssymb,amsfonts}
\usepackage{dsfont}
\newcommand{\mathbbm}[1]{\mathds{#1}}
\usepackage{multirow} 
\usepackage{tikz} 
\usepackage[normalem]{ulem} % for \sout
\usepackage{cite}
\usepackage{hyperref}
\usepackage[capitalize,noabbrev]{cleveref}
\usepackage{booktabs}
\usepackage[most]{tcolorbox}
\tcbset{
  colback=gray!6,
  colframe=black!35,
  boxrule=0.5pt,
  arc=2pt,
  left=6pt,right=6pt,top=4pt,bottom=4pt,
  boxsep=0pt,
}

\newtcolorbox{takeawaybox}[1]{%
  enhanced,
  arc=2pt,
  boxrule=0.4pt,
  colback=gray!6,
  colframe=black!35,
  left=3pt,right=3pt,top=2pt,bottom=2pt,
  boxsep=0pt,
  toptitle=0pt,bottomtitle=0pt,
  title=\textbf{Takeaway #1},
  fonttitle=\normalsize,
  coltitle=black,
  before skip=2pt,
  after skip=2pt,
}

\Crefname{algocf}{Algorithm}{Algorithms}
\Crefname{section}{\S}{\S\S}

\hypersetup{
    colorlinks=true,
    linkcolor=blue,
    citecolor=blue,
    urlcolor=blue,
}

\setlist[enumerate]{topsep=0pt,itemsep=0ex,partopsep=1ex,parsep=1ex}
\newtheorem{theorem}{Theorem}[section]   
\newtheorem{lemma}[theorem]{Lemma}

\crefname{lemma}{Lemma}{Lemmas}
\Crefname{lemma}{Lemma}{Lemmas}

\begin{document}

\title{Cache-Aware I/O Cost Modeling for Disk-Based Learned Indexes
}

\author{\IEEEauthorblockN{Zhanwei Shi$^1$, Meng Zhang$^2$, Guangyi Zhang$^3$, Sha Hu$^1$, Jianwei Liao$^1$, Jingshu Peng$^4$, Qiyu Liu$^1$, Yingxia Shao$^5$}
\IEEEauthorblockA{
\textit{$^1$Southwest University, $^2$Nanyang Technological University, $^3$SZTU, $^4$HKUST, $^5$BUPT}\\
\{zwshi.cs, guangyizhang.jan, liaotoad, qyliu.cs\}@gmail.com, meng.zhang@ntu.edu.sg, husha@swu.edu.cn, jpengab@cse.ust.hk
}
}

\maketitle

\begin{abstract}
Learned indexes have shown attractive space-time trade-offs in main-memory settings, yet a principled I/O cost model for their disk-resident deployments is still missing, which is a prerequisite for index tuning and query optimization.
The practically employed page buffer makes the problem even harder: under typical cache policies, many of the logical page references issued by the index are served by the buffer rather than reaching disk, so the effective physical I/O depends jointly on the workload, the cache policy, and the index configuration.
In this paper, we propose CAM, the \textit{first} cache-aware I/O cost model for learned indexes that takes practical cache eviction policies into consideration. 
CAM is not tied to a particular learned index design: 
it estimates page access distributions without full trace replay for mainstream learned index designs, and then combines them with I/O cost models to estimate effective physical I/Os. 
This formulation enables principled knob tuning by explicitly modeling the trade-off between index footprint and buffer capacity. 
We instantiate CAM for disk-based PGM-index and RMI, and further apply the same modeling principle to learned-index-based joins through a hybrid strategy that adaptively chooses point or range probes based on local key density.
Extensive experiments on real benchmarks show that CAM provides \textit{accurate and efficient} I/O estimation across diverse workloads: 
CAM-guided tuning improves PGM throughput by \textbf{1.17$\times$} over multicriteria PGM tuning and improves RMI throughput by \textbf{1.66$\times$} over CDFShop with I/O-related considerations. 
For learned-index-based joins, our hybrid strategy improves end-to-end performance by up to \textbf{8.8$\times$} over disk-based index nested-loop join. 
\end{abstract}


\vspace{-2ex}
\section{Introduction}\label{sec:introduction}

Efficient indexing is fundamental to modern database systems and large-scale data analytics.
Traditional index structures such as B$^+$-trees~\cite{DBLP:journals/acta/BayerM72,DBLP:journals/pvldb/ChenLFWS20,DBLP:conf/sigmod/RaoR00} provide robust performance across diverse workloads. Their page-oriented node layout makes I/O behavior \emph{regular and predictable}, and mature cost models have long been used by database optimizers to estimate access cost.
In contrast, \emph{learned indexes}~\cite{DBLP:conf/sigmod/KraskaBCDP18,DBLP:journals/pvldb/FerraginaV20,DBLP:conf/sigmod/DingMYWDLZCGKLK20,DBLP:journals/pvldb/LuDLMW21,DBLP:journals/pvldb/WuZCCWX21,DBLP:journals/pvldb/0006CJLXWLWZZWR22} replace tree nodes with compact ML models that directly approximate the mapping from keys to sorted positions.
By exploiting the key distribution, learned indexes reduce memory footprint and have demonstrated attractive space-time trade-offs over B$^+$-trees in \emph{in-memory} benchmarks~\cite{DBLP:journals/pvldb/MarcusKRSMK0K20,DBLP:journals/pvldb/SunZL23,DBLP:journals/pvldb/WongkhamLLZLW22}.

When learned indexes move to disk, two coupled modeling challenges emerge that prior work has not addressed jointly.
\emph{(C1) Irregular last-mile I/O.}
A disk-based learned index~\cite{DBLP:journals/pvldb/FerraginaV20,DBLP:journals/pacmmod/ZhangSZ24,DBLP:conf/osdi/DaiXGAKAA20} returns an approximate position rather than a page pointer. The correctness is typically recovered by a bounded last-mile search window that can span multiple pages depending on the prediction error, the page size, the position of the prediction inside a page, and the page-fetching strategy~\cite{DBLP:journals/pacmmod/ZhangSZ24}.
Unlike a B$^+$-tree node access, the resulting page-access pattern is irregular and tightly coupled to the index's internal geometry.
\emph{(C2) Physical I/O is absorbed by the buffer.}
Practical database systems buffer pages in memory and serve many logical references without physical disk access.
Classical I/O models such as DAM~\cite{DBLP:journals/cacm/AggarwalV88} count only block transfers, treating DRAM as free;
extensions such as the Affine model~\cite{DBLP:journals/topc/BenderCFJJJKMMP21} and parallel/parametric I/O models~\cite{DBLP:journals/topc/BenderCFJJJKMMP21,DBLP:conf/damon/PaponA21} capture sequentiality and storage-level concurrency, but none model the page buffer.
The two challenges interact: how many logical references the buffer eliminates depends on the same workload locality, model prediction error $\varepsilon$, and page layout that determine the last-mile window in (C1).

\begin{figure}[t]
    \centering
    \includegraphics[width=0.83\linewidth]{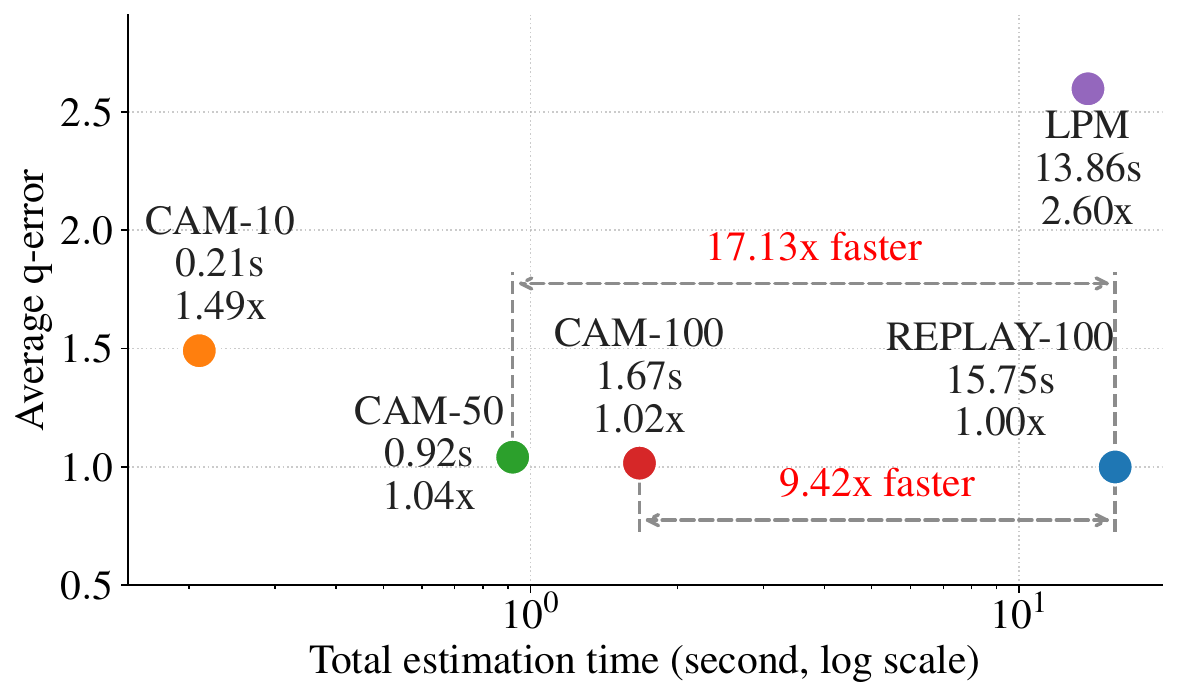}
    \caption{I/O cost estimation methods on a disk-based PGM index (books dataset, 1\,M point lookups, 128\,MiB LRU buffer). \textsc{LPM} counts all pages touched by the last-mile search; \textsc{Replay-100} executes the full query trace under the actual buffer; \textsc{CAM-}$x$ (this work) estimates from an $x$\% workload sample. CAM matches Replay accuracy but is \textbf{17.13}$\times$ faster; LPM is up to 2.6$\times$ off.}
    \label{fig:perf_overview}
\end{figure}

A naive workaround is to ignore the buffer and predict effective I/O directly from logical page counts (we call this the \emph{Logical Page Model}, LPM).
As shown in \cref{fig:perf_overview}, LPM mis-estimates physical I/O by up to \textbf{2.6$\times$} in Q-error against ground-truth trace replay on a disk-resident PGM-index, leading to poor index tuning and suboptimal query plans.
The accurate alternative is to fully simulate cache behavior by replaying the workload under a chosen replacement policy~\cite{DBLP:journals/pvldb/LeeAL23,DBLP:journals/pacmmod/LeisA0L023}, an approach widely used in trace-driven cache analysis~\cite{DBLP:conf/fast/MegiddoM03,DBLP:conf/cikm/OnLHWLX10}.
However, a single 1\,M-query replay takes tens of seconds, and the cost multiplies in tuning loops that sweep many index configurations.
What is needed is an estimator that is as accurate as replay but cheap enough to invoke inside the index tuner and query optimizer.

\noindent
\textit{\textbf{Our Approach.}}
We propose \textbf{CAM}, a \emph{Cache-Aware I/O cost Model} that derives the page-reference distribution \emph{structurally} from the geometry of the learned index, including window width, in-page offset, and page-fetching strategy, and composes it with a closed-form cache hit-rate estimation model (e.g., via Che's approximation~\cite{DBLP:journals/jsac/CheTW02}).
CAM has three design properties.
\emph{(i) Index-agnostic.}
The same formulation instantiates for error-bounded indexes such as PGM~\cite{DBLP:journals/pvldb/FerraginaV20} and model-routed indexes such as RMI~\cite{DBLP:conf/sigmod/KraskaBCDP18} (see \cref{fig:learned_index} for their structural difference).
\emph{(ii) Policy-pluggable.}
We provide hit-rate estimators for a wide spectrum of cache policies, including FIFO, LRU, and LFU, plus a closed-form result for sorted workloads that is policy-independent.
\emph{(iii) Compositional with device-side models.}
CAM's output (i.e., the effective number of physical I/Os) can be integrated with DAM, the Affine model, or parametric I/O models.
Notably, modeling only the physical I/O term is appropriate because, on disk-resident learned indexes, the index itself fits in memory and traversal is as cheap as $O(\log\log n)$~\cite{DBLP:journals/pvldb/LiuHQPLLC25}, so query latency is typically I/O-dominated (over \textbf{90.2\%}; see \cref{sec:pre}).

\noindent
\textit{\textbf{What CAM Enables.}}
Because CAM exposes the \emph{index footprint vs.\ buffer capacity} trade-off analytically under a memory budget $M$, it unlocks two optimizations that prior cache-oblivious learned index tuners~\cite{DBLP:journals/pvldb/FerraginaV20,DBLP:conf/sigmod/MarcusZK20} cannot express.
First, for size--accuracy index tuning, CAM reduces the problem to a single-objective search: pick the parameter $\theta^\star$ (e.g., the error bound $\varepsilon$ of PGM or the branching factor of RMI) that minimizes estimated effective I/O subject to $M$.
Second, for learned-index-based joins, CAM enables a \emph{hybrid} probe strategy~\cite{DBLP:conf/acda/Chesetti025}: dense probe regions favor range probes, which amortize per-key traversal and buffer overhead, while sparse regions favor point lookups, which avoid fetching irrelevant pages.
Guided by CAM, we partition the sorted probe sequence and select the cheaper strategy per region.

% Instantiating CAM is nontrivial across query types and we develop a dedicated estimator for each.
% RMI lacks a global error bound, so its page-reference distribution is a workload-weighted mixture over leaves.
% Range queries couple two endpoint errors and their joint window must be modeled directly.
% Sorted join workloads violate the independent-reference assumption underlying Che's approximation, but admit a tight closed form that is policy-independent.

\noindent
\textit{\textbf{Summary of Contributions.}}
\begin{enumerate}[noitemsep,leftmargin=*,label={\bf C\arabic*:}]
    \item \textbf{CAM I/O cost model (\cref{sec:cam}).}
    We introduce the first analytical, replay-free cache-aware I/O cost model for disk-based learned indexes under self-managed page buffers.
    CAM factors cost into a structural page-reference distribution times a policy-specific cache hit-rate model, and composes with existing device-side I/O abstractions.

    \item \textbf{Workload-aware page-reference estimators (\cref{sec:h-estimate}).}
    We derive efficient estimators for point, range, and join workloads, including a closed-form result for sorted access workloads (\cref{thm:ordered_hit_rate}) that is policy-independent.
    CAM matches replay accuracy (\textbf{1.04$\times$} Q-error) while reducing estimation time by \textbf{17.13$\times$}.

    \item \textbf{Memory-budgeted index tuning (\cref{sec:application}).}
    We use CAM to formulate learned-index tuning as a single-objective problem under a fixed memory budget, explicitly trading off index footprint and buffer capacity.
    CAM-tuned PGM improves throughput by up to \textbf{1.17$\times$} over the multicriteria PGM optimizer~\cite{DBLP:journals/pvldb/FerraginaV20};
    CAM-tuned RMI improves throughput by up to \textbf{1.66$\times$} over CDFShop~\cite{DBLP:conf/sigmod/MarcusZK20}, while reducing tuning time by \textbf{75.7\%} and \textbf{60.1\%}, respectively.

    \item \textbf{CAM-guided hybrid join (\cref{sec:join-query}).}
    We design a density-aware hybrid join that partitions the sorted probe sequence and selects point or range probing per region from CAM's estimate.
    In our join evaluation across representative workload mixtures, the hybrid strategy improves end-to-end join time by up to \textbf{8.80$\times$} over conventional index nested-loop join.
\end{enumerate}

\noindent
\textit{\textbf{Paper Organization.}}
\cref{sec:pre} reviews learned indexes, page-fetching strategies, and cache policies.
% \cref{sec:cam} develops the CAM I/O cost model and its per-policy hit-rate components.
\cref{sec:cam} presents the CAM I/O cost model and formulates its policy-specific hit-rate and data-access cost components.
\cref{sec:h-estimate} derives the page-reference estimators for point, range, and join queries.
\cref{sec:application} applies CAM to memory-budgeted PGM and RMI tuning, and \cref{sec:join-query} introduces the CAM-guided hybrid join strategy.
\cref{sec:exp} reports the experimental evaluation, and \cref{sec:conclusion} concludes.

\begin{figure}[t]
    \centering
    \includegraphics[width=0.38\textwidth]{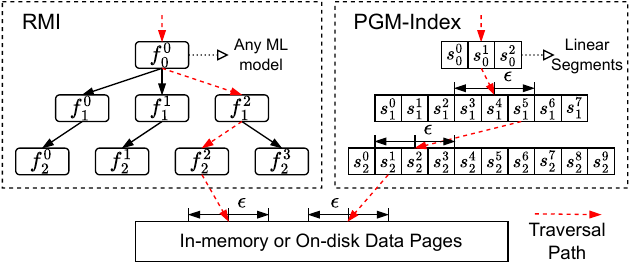}
    \caption{Illustration of model-routed learned index (e.g., RMI, left) and error-bounded learned index (e.g., PGM, right). The red dashed arrows denote the index traversal path predicted by ML models.}
    \label{fig:learned_index}
\end{figure}

\section{Background and Related Work}\label{sec:pre}
In this section, we review the basis of learned indexes, page fetching strategies, and cache replacement policies.

\subsection{Learned Index}
\noindent
\textit{\textbf{Definitions.}} 
Formally, given a sorted set of integer keys $\mathcal{K}=\{k_1,\ldots,k_n\}$ and a predefined error bound $\varepsilon$, learned indexes approximate the rank function $r: k_i\mapsto i$ using an ML model $f(k)$ such that $|\lfloor f(k)\rfloor-r(k)|\le\varepsilon$ for $\forall i\in\{1,\cdots,n\}$. 
At query time, the model prediction defines a bounded last-mile search interval $[\lfloor f(k)\rfloor-\varepsilon, \lfloor f(k)\rfloor+\varepsilon]$, which can be resolved using binary or exponential search.
Representative learned indexes, including RMI~\cite{DBLP:conf/sigmod/KraskaBCDP18}, RadixSpline~\cite{DBLP:conf/sigmod/KipfMRSKK020}, PGM~\cite{DBLP:journals/pvldb/FerraginaV20}, ALEX~\cite{DBLP:conf/sigmod/DingMYWDLZCGKLK20}, LIPP~\cite{DBLP:journals/pvldb/WuZCCWX21}, and DILI~\cite{DBLP:journals/pvldb/LiLZDYP23} have demonstrated superior space-time trade-offs over conventional B$^+$-trees in recent benchmark studies~\cite{DBLP:journals/pvldb/MarcusKRSMK0K20,DBLP:journals/pvldb/SunZL23,DBLP:journals/pvldb/WongkhamLLZLW22}.

\noindent
\textit{\textbf{RMI vs.~PGM.}}
\Cref{fig:learned_index} illustrates two typical learned-index designs: RMI (model-based routing) and PGM (error-bounded search).
Specifically, RMI~\cite{DBLP:conf/sigmod/KraskaBCDP18} organizes models in a hierarchy and routes a search key to a leaf model that predicts its position. 
While widely studied and tuned~\cite{DBLP:conf/sigmod/MarcusZK20,DBLP:journals/pvldb/MaltryD22,DBLP:journals/pvldb/MarcusKRSMK0K20}, RMI's leaf-level error depends on the entire root-to-leaf path, making it difficult to characterize analytically. 
In contrast, PGM~\cite{DBLP:journals/pvldb/FerraginaV20} adopts error-bounded piecewise linear approximation ($\varepsilon$-PLA) to recursively fit the key space until a single line segment is reached, making the leaf-level error fixed to $\varepsilon$. 

Although prior work shows that model-routed learned indexes often outperform PGM under in-memory workloads~\cite{DBLP:journals/pvldb/MarcusKRSMK0K20,DBLP:journals/pvldb/SunZL23,DBLP:journals/pvldb/WongkhamLLZLW22}, this advantage does not carry over to the out-of-memory setting. 
Due to the lack of explicit control over last-mile search error, RMI incurs substantially higher I/O costs (\cref{fig:query_breakdown}) than a comparably sized PGM on the osm dataset which exhibits weak local structure~\cite{DBLP:journals/pvldb/MarcusKRSMK0K20}, especially in the tail (\cref{tab:PGM_vs_RMI}).
Moreover, PGM's hard error constraint makes its I/O behavior easier to analyze, whereas RMI's per-leaf error depends on complex model choices and training dynamics.
For these reasons, although CAM applies to both PGM and RMI tuning (\cref{sec:application}), it is more precise on PGM than on RMI.

\begin{table}[t]
  \caption{Evaluation of I/O size (in bytes) for 1M random key lookup queries w4 using SSD-based PGM and RMI under comparable index sizes on the osm dataset. 
  }
  \label{tab:PGM_vs_RMI}
  \centering
  \scriptsize
  \begin{tabular}{ccccccc}
    \toprule
    \textbf{Index Type} & \textbf{Mean} & \textbf{Std.} & \textbf{50th pctl} & \textbf{95th pctl} \\
    \midrule
    PGM & 4937.2 & 1654.7 & 4096.0 & 8192.0 \\
    RMI & 194882.5 & 565090.7 & 57344.0 & 733184.0 \\
    \bottomrule
  \end{tabular}
  \vspace{-2ex}
\end{table}

\begin{figure}[t]
    \centering
    \begin{subfigure}{0.75\linewidth}
        \centering
        \includegraphics[width=\linewidth]{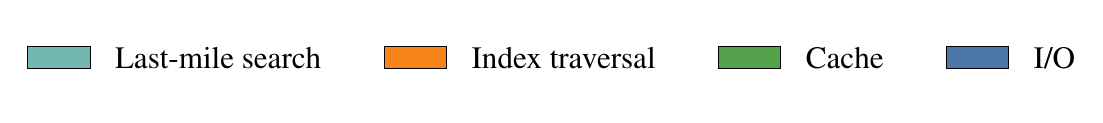}
    \end{subfigure}
    \begin{subfigure}{0.45\linewidth}
        \centering
        \includegraphics[width=\linewidth]{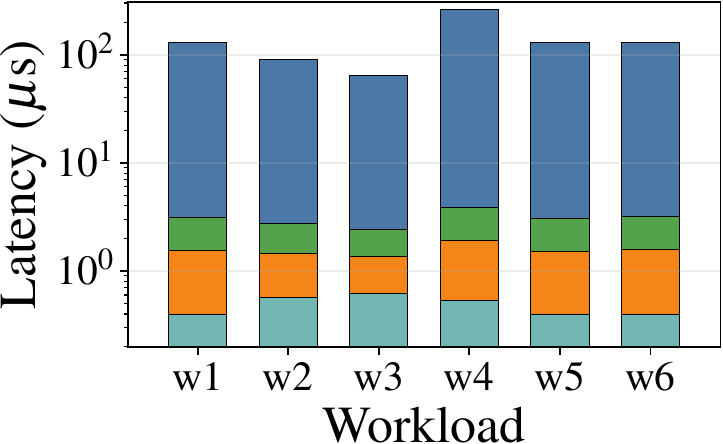}
        \caption{PGM}
        \label{fig:pgm-io-lat}
    \end{subfigure}
    \hfill
    \begin{subfigure}{0.45\linewidth}
        \centering
        \includegraphics[width=\linewidth]{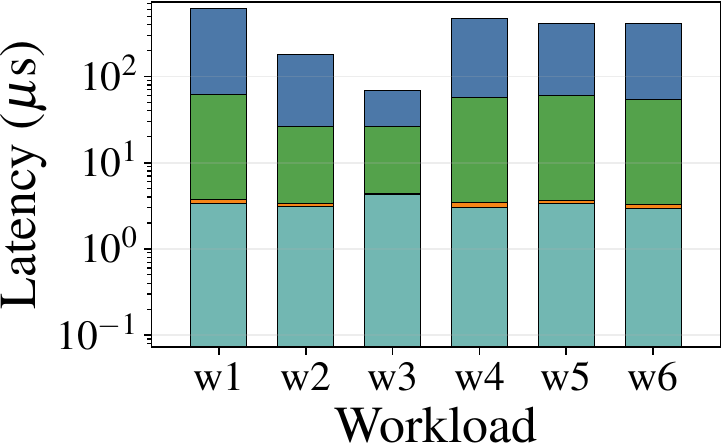}
        \caption{RMI}
        \label{fig:rmi-io-lat}
    \end{subfigure}
    \caption{Query latency (log-scale) breakdown of PGM and RMI across representative workloads (see \cref{tab:workload_mixture}) on the osm dataset with 128\,MiB LRU page buffer. The two indexes are configured with comparable index sizes for a fair comparison. Even with modern high-performance SSDs and OS-level page buffering, I/O remains the dominant component of total query latency, justifying CAM's focus on modeling the I/O term.}
    \label{fig:query_breakdown}
\end{figure}

\begin{figure}[h]
    \centering
    \includegraphics[width=\linewidth]{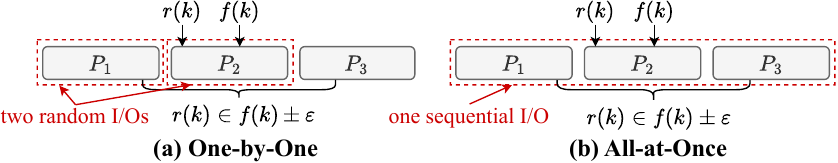}
    \caption{Page fetching strategies for learned index.} 
    % \gz{Can we make the left clear: first fetch P2, and then fetch P1?}\szw{One-by-one first fetches the lower-bound page $P_1$ and fetches $P_2$ only after checking that $k$ exceeds the maximum key in $P_1$.}
    \label{fig:leaf-fetch-strategy}
    \vspace{-2ex}
\end{figure}

\subsection{Page Fetching Strategy}\label{subsec:page_fetching}
\noindent
\textit{\textbf{Leaf-Page Fetching Strategy for Disk-based Learned Index.}}
Following prior experience on I/O-aware learned index~\cite{DBLP:journals/pacmmod/ZhangSZ24,DBLP:conf/icde/LanBCBD24,DBLP:journals/pacmmod/LanBCB23}, 
we adopt an index-data separation design: 
data pages are stored in sorted order on SSD, while the (typically compact~\cite{DBLP:conf/icml/FerraginaLV20H,DBLP:journals/pvldb/LiuHQPLLC25}) index part resides entirely in memory. 
In this setting, a learned index can be treated as a black-box oracle $f$ that, for a search key $k$, returns an error-bounded search window $I_k=[f(k)-\varepsilon, f(k)+\varepsilon]$ using negligible CPU overhead. 
Prior work~\cite{DBLP:journals/pacmmod/ZhangSZ24} describes two ways to translate $I_k$ into on-disk page fetching requests (illustrated in \cref{fig:leaf-fetch-strategy}). 
\begin{enumerate}[noitemsep,leftmargin=*,label={\bf S\arabic*:}]
    \item \textbf{One-by-one fetching} probes pages incrementally around $f(k)$ and issues the next read only after observing the previous page, potentially incurring multiple dependent random I/Os.
    \item \textbf{All-at-once fetching} reads all pages overlapping the interval in a single batched/coalesced request, yielding one larger sequential I/O.
\end{enumerate}

We evaluate these two page-fetching strategies by issuing random lookup queries while varying the index error bound $\varepsilon$ and the number of threads. 
As shown in \cref{fig:leaf-fetch-strategy-comparison}, in most cases where $\varepsilon\leq4096$ and thread count $\geq16$, the all-at-once strategy and one-by-one fetching exhibit comparable performance, consistent with prior observations~\cite{DBLP:journals/pacmmod/ZhangSZ24}. However, while one-by-one fetching can potentially reduce the total number of page reads, it typically triggers a chain of dependent random I/Os, which limits the opportunity to utilize parallelism of modern SSDs and I/O interfaces such as io\_uring. 
Given these observations, we theoretically analyze and compare both page fetching strategies in \cref{sec:cam}, and we adopt the all-at-once strategy as the default approach in subsequent experiments.

\begin{figure}[t]
    \centering
    \includegraphics[width=0.7\linewidth]{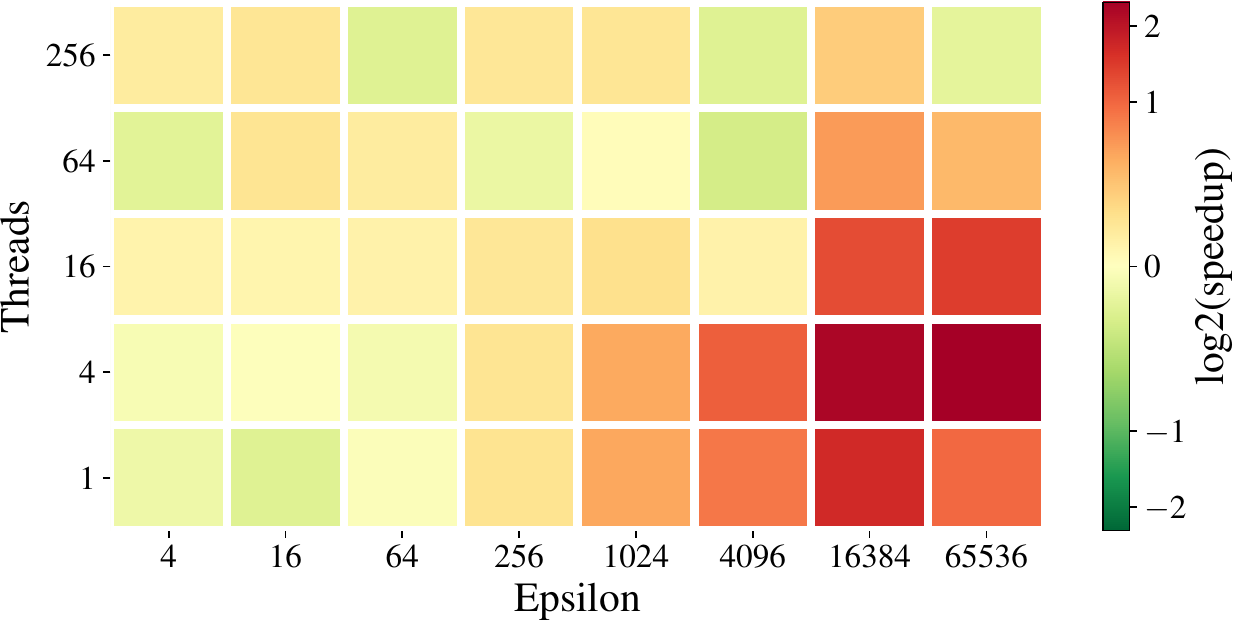}
    \caption{Throughput speedup (All-at-Once over One-by-One) evaluation by varying thread count and $\epsilon$. Note that the speedup ratio is on a logarithmic scale.}
    \label{fig:leaf-fetch-strategy-comparison}
\end{figure}

% \begin{figure}[h]
% \centering
%     \includegraphics[width=\linewidth]{figures/cache_policy_legend.pdf}
%     \begin{subfigure}[b]{0.48\linewidth}
%         \centering
%         \includegraphics[width=\linewidth]{figures/cache_policy_simple_crop.pdf}
%         \caption{Random workload.}
%         \label{fig:cache_policy_simple}
%     \end{subfigure}
%     \begin{subfigure}[b]{0.48\linewidth}
%         \centering
%         \includegraphics[width=\linewidth]{figures/cache_policy_sorted_crop.pdf}
%         \caption{Sorted workload.}
%         \label{fig:cache_policy_sorted}
%     \end{subfigure}
%     \caption{Evaluation of three cache replacement policies (LRU/FIFO/LFU) under varying buffer sizes. 
%     We execute 1M point queries and measure the average end-to-end query latency (bars) and cache hit ratio (lines). }
%     \label{fig:cache_performance}
%     \vspace{-2ex}
% \end{figure}

\subsection{Page Buffer and Eviction Policy}
\noindent
\textit{\textbf{Role of Page Buffer.}}
Disk-based DBMSs typically maintain a dedicated page buffer in user space, even though the OS concurrently maintains a filesystem buffer. 
The OS buffer is a system-wide cache of file blocks managed by the kernel, whereas the DBMS buffer is database-aware: 
it tracks logical page identifiers, enforces pinning/latching for concurrent operators, and coordinates dirty-page writeback. 
Thus, relying on the OS alone is vulnerable to uncoordinated eviction and degraded performance~\cite{DBLP:journals/cacm/Stonebraker81,DBLP:journals/pacmmod/LeisA0L023}. 

\noindent
\textit{\textbf{Eviction Policies.}}
The eviction policy determines which pages to evict under capacity pressure, directly affecting buffer hit rate and query performance. 
We consider three typical strategies. 
(1) \textit{FIFO} evicts the oldest resident page in arrival order; it is simple but can suffer from common access patterns like cyclic patterns~\cite{DBLP:journals/cacm/BeladyNS69}. 
(2) \textit{LFU} retains pages with high historical reference counts, yet may hoard stale but once-frequent pages and adapt slowly to workload shifts and phase changes~\cite{DBLP:conf/hotstorage/VietriRMLLRZN18}. 
(3) \textit{LRU} evicts the least recently used page, offering a lightweight, recency-based approach commonly used as a baseline in DBMS buffers~\cite{DBLP:conf/sigmod/ONeilOW93,DBLP:conf/vldb/JohnsonS94,DBLP:conf/sigmetrics/JiangZ02,DBLP:conf/fast/MegiddoM03,DBLP:conf/cases/ParkJKKL06}. 
% Beyond these baselines, a large body of work proposes more robust policies for scans and mixed workloads (e.g., LRU-K, 2Q, LIRS, ARC, CFLRU). 

% We evaluate each eviction policy using page-buffer traces generated by PGM-Index lookups on the OSM dataset~\cite{DBLP:conf/sigmod/KraskaBCDP18}. 
% The \textit{random workload} draws 1M query keys from a mixed distribution: 40\% hotspots (a few local regions following Zipf(1.5)), 30\% Zipf(1.2), and 30\% uniform; the \textit{sorted workload} reorders the same queries by key. 
% As shown in \cref{fig:cache_policy_simple,fig:cache_policy_sorted}, under the random workload, FIFO achieves the lowest hit ratio, suggesting that it fails to exploit locality and thus incurs unnecessary evictions. 
% Under the sorted workload, LFU performs worst when the buffer is small because frequency estimates are dominated by transient bursts and early noise, which ``locks in'' a few items, hurting short-term locality. 
% Given its consistently high and stable hit ratio across cache sizes and workloads, we adopt LRU as the default eviction policy in our framework. 
% Building on LRU, our work is the first to account for the I/O cost incurred by a DBMS-managed page buffer when a learned index processes different query types (\cref{sec:cam}). 

\section{Cache-Aware I/O-Cost Modeling}\label{sec:cam}
{We introduce Cache-Aware I/O Model (CAM), the first I/O cost model for disk-based learned indexes that estimates effective physical I/O by modeling how a self-managed page buffer serves logical page references.}

\begin{table}[t]
\caption{Relative contribution ratio (r) of covariance to expected I/O 
($\mathbb{E}[\mathrm{IO}]$) across cache policies, 
error bounds ($\varepsilon$), and memory budgets.}
\label{tab:correlation-between-terms}
\scriptsize
\renewcommand{\arraystretch}{1.15}
\setlength{\tabcolsep}{2.5pt}
\centering
\begin{tabular}{cc cc cc cc}
\toprule
\multicolumn{2}{c}{}
& \multicolumn{2}{c}{$\varepsilon=8$}
& \multicolumn{2}{c}{$\varepsilon=16$}
& \multicolumn{2}{c}{$\varepsilon=64$} \\
\cmidrule(lr){3-4}\cmidrule(lr){5-6}\cmidrule(lr){7-8}
\textbf{Policy} & \textbf{Memory}
& \textbf{$\mathbb{E}(\mathrm{IO})$} & r (\%)
& \textbf{$\mathbb{E}(\mathrm{IO})$} & r (\%)
& \textbf{$\mathbb{E}(\mathrm{IO})$} & r (\%) \\
\midrule

\multirow{3}{*}{FIFO}
& 20 MB & 2.962 & -1.925 & 2.804 & -3.594 & 3.092 & -3.009 \\
& 40 MB & 2.025 & -1.900 & 1.869 & -3.576 & 2.009 & -2.969 \\
& 60 MB & 1.178 & -1.927 & 1.007 & -3.573 & 1.000 & -3.042 \\
\midrule

\multirow{3}{*}{LRU}
& 20 MB & 2.899 & -1.932 & 2.748 & -3.620 & 3.031 & -3.008 \\
& 40 MB & 1.987 & -1.915 & 1.836 & -3.584 & 1.971 & -2.972 \\
& 60 MB & 1.156 & -1.889 & 0.993 & -3.549 & 0.982 & -2.993 \\
\midrule

\multirow{3}{*}{LFU}
& 20 MB & 2.838 & -1.937 & 2.703 & -3.641 & 2.986 & -3.011 \\
& 40 MB & 1.966 & -1.946 & 1.819 & -3.652 & 1.955 & -3.072 \\
& 60 MB & 1.147 & -1.898 & 0.985 & -3.606 & 0.979 & -2.957 \\
\bottomrule
\end{tabular}
\end{table}

\subsection{Cache-Aware I/O Cost Estimation}
\noindent
\textit{\textbf{Existing OS-Level I/O Cost Models.}}
Modern storage systems are deeply hierarchical, and I/O latency typically dominates CPU and DRAM costs. 
The Disk Access Machine (DAM)~\cite{DBLP:journals/cacm/AggarwalV88} only counts the number of block transfers, treating CPU and DRAM as free. 
The Affine Model~\cite{DBLP:journals/topc/BenderCFJJJKMMP21} refines DAM by modeling the cost of an I/O of size $x$ as $1+\alpha x$, where $1$ represents a normalized per-I/O setup overhead and $\alpha$ is the per-unit transfer cost. 
The PDAM model~\cite{DBLP:journals/topc/BenderCFJJJKMMP21} additionally parameterizes device-level I/O parallelism, while the Parametric I/O Model (PIO)~\cite{DBLP:conf/damon/PaponA21} further accounts for device-specific read/write concurrency and read-write asymmetry.

\noindent
\textit{\textbf{Introduction to CAM.}} 
In practice, many logical page requests never reach the device because they are served from the in-memory page buffer. CAM captures this effect by predicting the buffer hit rate and translating it into an \textit{effective} number of physical I/Os. This buffer-aware I/O size can be composed with \textit{any} OS-level model (e.g., DAM or the Affine model) to estimate I/O costs. 
In this sense, CAM serves as a cache-aware preprocessor that enables existing I/O models to more accurately reflect the behavior of memory-constrained learned-index workloads.

For a query $Q$, let $\mathrm{DAC}(Q)$ (data access cost) denote the number of logical page requests issued by the last-mile search of learned index, and let $H(Q)$ denote the fraction of these requests served by the in-memory page buffer. 
The physical I/Os incurred are
\begin{equation}
    \mathrm{IO}(Q)=(1-H(Q))\cdot\mathrm{DAC}(Q).
\end{equation}
% Taking expectation over the query distribution and the stationary cache state yields
Taking expectation over the query distribution yields:
\begin{equation}\label{eq:expected_io}
    \mathbb{E}[\mathrm{IO}]= 
    (1-\mathbb{E}[H])\cdot\mathbb{E}[\mathrm{DAC}]-\operatorname{Cov}(H,\mathrm{DAC}).
\end{equation}
$\mathrm{DAC}(Q)$ is primarily determined by the width of the learned index search window, which further depends on the maximum model prediction error $\varepsilon$, while $H(Q)$ is mainly governed by long-run page popularity and buffer capacity. 
Since these two factors arise from different mechanisms, the covariance term is typically small. 
As shown in \cref{tab:correlation-between-terms},
across all cache policies, error bounds, and memory budgets from our experiments, the covariance term contributes at most $3.7\%$ in magnitude relative to $\mathbb{E}[\mathrm{IO}]$, confirming its negligible impact. 
Therefore, CAM approximates the expected physical I/O cost by
\begin{equation}    \label{eq:I/O_cost_model}
\text{Cost}_{\text{CAM}}
\approx
(1-h)\cdot \mathbb{E}[\mathrm{DAC}],
\end{equation}
where $h=\mathbb{E}[H]$ can be estimated for different cache policies.

\subsection{Hit Rate for Different Cache Policies}
\label{subsec:hit_rate}

For unordered workloads, we estimate buffer hit rates under the Independent Reference Model (IRM), where logical page requests are independently drawn from a static, time-invariant popularity distribution. Let $\Pr\nolimits_{\mathrm{req}}(i)$ denote the probability that a randomly selected logical page request refers to page $i$, $C$ the cache capacity in pages, and $N$ the number of distinct pages touched by the last-mile windows across all queries.

\noindent
\textit{\textbf{FIFO.}} 
Under the IRM assumption, FIFO has been shown to have exactly the same hit rates as random eviction~\cite{DBLP:journals/tc/Gelenbe73a}. 
Let $h(i)$ denote the stationary probability that object $i$ is cached. 
Fricker~\cite{DBLP:conf/teletraffic/FrickerRR12} shows that 
\begin{equation}
    h(i) = \frac{\Pr\nolimits_{\text{req}}(i)\cdot\tau_C}{\sum_{x\neq i}\Pr\nolimits_{\text{req}}(x)+\Pr\nolimits_{\text{req}}(i)\cdot\tau_C},
\end{equation}
where $\tau_C$ is the cache characteristic time that captures the average residency opportunity.
$\tau_C$ can be determined by the cache-size consistency condition
\begin{equation} \label{eq:fifo_cache_consistency_equation}
    C = \sum_{i=1}^{N} \frac{\Pr\nolimits_{\text{req}}(i)\cdot\tau_C}{1-\Pr\nolimits_{\text{req}}(i)+\Pr\nolimits_{\text{req}}(i)\cdot\tau_C}.
\end{equation}
Finally, the expected hit rate is computed as
\begin{equation}\label{eq:FIFO_hit_rates}
    h_{\mathrm{FIFO}}  = \sum_{i=1}^{N} \Pr\nolimits_{\text{req}}(i)\cdot h(i).
\end{equation}
% When the buffer can hold all distinct pages (i.e., $C\geq N$), $\tau_C$ cannot be derived from \cref{eq:fifo_cache_consistency_equation}. 
% In this case, each page incurs at most one compulsory miss: only the first access triggers an I/O, and the hit rate becomes $h=(R-N)/R$ ($R$ is the total number of page requests). 

\noindent
\textit{\textbf{LRU.}} 
Under the IRM assumption, the hit rate of an LRU-managed buffer can be effectively approximated via \textit{Che's approximation}~\cite{DBLP:journals/jsac/CheTW02,DBLP:journals/tompecs/GarettoLM16}:
 \begin{equation}\label{eq:che_approximation_func}
     h_{\mathrm{LRU}} = \sum_{i=1}^{N} \Pr\nolimits_{\text{req}}(i)\cdot \left(1 - \exp \left(-\Pr\nolimits_{\text{req}}(i) \cdot T_C\right)\right),
 \end{equation}
 where $\Pr_{\text{req}}(i)$ is the request probability of page $i$, and the characteristic time $T_C$ can be similarly determined by solving the cache-size consistency equation: 
 \begin{equation} \label{eq:cache_consistency_equation}
  C = \sum_{i=1}^{N} \left(1 - \exp\left(-\Pr\nolimits_{\text{req}}(i)\cdot T_C\right) \right).
 \end{equation}
% When $C\geq N$, the same compulsory-miss argument as in the FIFO case applies, and the hit rate reduces to $h=(R-N)/R$.

% The constant $T_{C}$ can be numerically computed by solving \cref{eq:cache_consistency_equation} with Brent's method~\cite{DBLP:journals/cj/Brent71}. 
% When the buffer can hold all distinct pages (i.e., $C\geq N$), \cref{eq:cache_consistency_equation} has no finite root (i.e., $T_C\to\infty$). 
% In this case, each page incurs at most one compulsory miss: only the first access triggers an I/O, and the hit rate becomes $h=(R-N)/R$ ($R$ is the total number of page requests). 

\noindent
\textit{\textbf{LFU.}} 
Under the IRM assumption, LFU admits the simplest steady-state characterization. As analyzed in~\cite{DBLP:journals/cn/HasslingerONHH23}, once the buffer has converged, LFU keeps the $C$ most popular pages, and its hit rate equals the cumulative request probability of those top-$C$ pages. 
Let $p_i=\Pr_{\mathrm{req}}(i)$, and let $p_{(1)}\ge p_{(2)}\ge \cdots \ge p_{(N)}$ denote the sorted page request probabilities. Then
\begin{equation}    \label{eq:LFU_hit_rates}
    h_{\mathrm{LFU}}=\sum_{i=1}^{C} p_{(i)}.
\end{equation}

When the buffer can hold all distinct referenced pages (i.e., $C\geq N$), we directly compute the hit rate from compulsory misses rather than solving the policy-specific equations. Each distinct page incurs only one compulsory miss on its first reference, and all subsequent references hit in the buffer. Thus, the hit rate reduces to $h=(R-N)/R$, where $R$ is the total number of page requests.
% For a finite query trace of length $Q$, directly applying \cref{eq:LFU_hit_rates} may overestimate the hit rate, since it ignores the warm-up phase before those hot pages become resident. In particular, each retained top-$C$ page still incurs one compulsory miss when first accessed. The probability that page $i$ is accessed at least once in the trace is $1-(1-p_{(i)})^Q$, yielding the following cold-start correction:
% \begin{equation}    \label{eq:LFU_hit_rates_cold}
%     h_{\mathrm{LFU}}^{(Q)}
%     \approx
%     \sum_{i=1}^{C} p_{(i)}
%     -
%     \frac{1}{Q}\sum_{i=1}^{C}\Bigl(1-(1-p_{(i)})^Q\Bigr).
% \end{equation}

\subsection{Hit Rate under Sorted Workloads}
The estimators in \cref{subsec:hit_rate} are derived under the IRM assumption~\cite{DBLP:journals/jsac/CheTW02}, but remain accurate for many popularity laws even with mild IRM violations~\cite{DBLP:conf/teletraffic/FrickerRR12}. 
However, for strictly ordered query key streams (common in join workloads discussed in \cref{sec:join-query}), the page references become highly correlated, making \cref{eq:FIFO_hit_rates,eq:che_approximation_func,eq:LFU_hit_rates} underestimate the hit rate as the temporal locality is ignored. 
Accordingly, CAM uses the estimator in \cref{thm:ordered_hit_rate} for sorted workloads.

\begin{theorem}[Hit Rate for Sorted Queries] \label{thm:ordered_hit_rate}
Let $\mathcal{K}=(k_1,\ldots,k_{|\mathcal K|})$ be an array of keys. 
A query sequence $\mathcal{Q}$ is \textbf{sorted} w.r.t.\ $\mathcal{K}$ if it requests $k_i$ before $k_j$ for all $i<j$.  
For any such $\mathcal{Q}$, suppose the buffer capacity $C$ satisfies $C\geq 1 + \lceil {2\varepsilon}/{C_{\textnormal{ipp}}} \rceil$,
where $C_\textnormal{ipp}$ is the number of items per page. 
Let $R$ denote the total number of page requests and $N$ the number of distinct pages referenced by the workload. 
Then the cache hit rate equals $h=(R-N)/R$, which is identical to the large-capacity case discussed in \cref{subsec:hit_rate}. 
\end{theorem}
\begin{proof}
Let $(p_1,\ldots,p_R)$ be the page reference sequence generated by processing the $m$ queries, and partition it by queries: $(p_1,\ldots,p_R)=\pi_1 \Vert \pi_2 \Vert \cdots \Vert \pi_m$,
where $\pi_t$ is the subsequence of page IDs referenced while processing query $t$. 
For a learned-index-based engine, query $t$ touches exactly the pages in its last-mile window $W_t=[L_t,H_t]$, with $|W_t|\le 1 + \lceil 2\varepsilon/C_{\text{ipp}} \rceil$.
Let $\mathcal{P}$ be the set of distinct pages appearing in $(p_1,\ldots,p_R)$, and $|\mathcal{P}|=N$. Since queries are sorted, the windows move monotonically, i.e., $L_{t+1}\ge L_t$. Hence, between two consecutive queries, only pages newly entering the window can miss: pages in $W_{t+1}\cap W_t$ are still resident and therefore hit, while pages in $W_{t+1}\setminus W_t$ may incur misses.
Because $C \ge |W_t|$ for all $t$, the entire window $W_t$ fits in cache during the processing of query $t$, so no page in $W_t$ can be evicted before $\pi_t$ finishes. By monotonicity of $L_t$, once a page is loaded it is either reused by subsequent overlapping windows (and thus hits) or never referenced again. Therefore, each distinct page in $\mathcal{P}$ incurs exactly one compulsory miss on its first reference and all later references are hits. The total number of misses is $N$, so the hit rate is
$h=(R-N)/R$, as claimed.
\end{proof}

\subsection{Expected Data Access Cost}
\label{subsec:dac}

We next derive analytical models for the expected logical page requests 
$\mathbb{E}[\mathrm{DAC}]$ (\cref{eq:I/O_cost_model}) under different 
page fetching strategies (see \cref{subsec:page_fetching}).

\begin{lemma}[DAC for All-at-Once Fetching]\label{lemma:dac_aao}
If the predicted position lies in a page with a uniformly distributed offset, the expected number of I/Os with all-at-once strategy is
\begin{equation}\label{eq:EDAC_for_all_in_once}
\mathbb{E}[\mathrm{DAC}] = 1 + \frac{2\varepsilon}{C_{\textnormal{ipp}}},
\end{equation}
where $C_{\textnormal{ipp}}$ represents the number of items per page.
\end{lemma}
\begin{proof}
Denote the offset within the page of the predicted position by $s\sim U[0,C_{ipp}-1]$. The I/O needs to fetch the page that contains the predicted position and additional pages required to cover the left and right of the $\pm\varepsilon$ window:
\[
    \mathbb{E}[\text{DAC}]
    =\frac{1}{C_{\text{ipp}}}\sum_{s=0}^{C_{\text{ipp}}-1}
    \left(
    1+\Big\lceil \tfrac{\varepsilon-s}{C_{\text{ipp}}} \Big\rceil
    +\Big\lceil \tfrac{\varepsilon-(C_{\text{ipp}}-1-s)}{C_{\text{ipp}}} \Big\rceil
    \right).
\]
Rewrite $\varepsilon=\lambda\cdot C_{\text{ipp}}+r$ for some $\lambda\in\mathbb{N}$ and $0\le r<C_{\text{ipp}}$. 
Then we have
\[
\begin{aligned}
    \mathbb{E}[\text{DAC}] &= \sum_{s=0}^{C_{\text{ipp}}-1} \frac{1+(\lambda+\mathbbm{1}_{\{s<r\}}) + (\lambda+\mathbbm{1}_{\{s>C_{ipp}-1-r\}})}{C_{ipp}}\\
    &= 1+2\lambda+\frac{2r}{C_{\text{ipp}}} = 1+\frac{2\varepsilon}{C_{\text{ipp}}},
\end{aligned}
\]
which completes the proof.
\end{proof}

\begin{lemma}[DAC for One-by-One Fetching]\label{lemma:dac_obo}
Under the same assumptions as Lemma~\ref{lemma:dac_aao}, the expected number of I/Os under the one-by-one strategy is
\begin{align}   \label{eq:EDAC_for_one_by_one}
\mathbb{E}[\text{DAC}] = 1+\frac{\varepsilon}{C_{\text{ipp}}}. 
\end{align}
\end{lemma}

\begin{proof}
Let $\hat y$ denote the position predicted by learned index.
Let $k\sim U[0,C_{\text{ipp}}-1]$ be the page offset of the lower bound 
$\hat y-\varepsilon$, and let $X\sim U[0,2\varepsilon]$ be the distance 
from that bound to the true position. Under one-by-one fetching, the 
number of additional pages beyond the first is 
$\lfloor(k+X)/C_{\text{ipp}}\rfloor$. Hence
\[
\mathbb{E}[\mathrm{DAC}]
=1+\frac{1}{(2\varepsilon+1)\,C_{\text{ipp}}}
\sum_{x=0}^{2\varepsilon}\sum_{k=0}^{C_{\text{ipp}}-1}
\Big\lfloor\frac{k+x}{C_{\text{ipp}}}\Big\rfloor .
\]
For any fixed $x$, write $x=q\cdot C_{\text{ipp}}+r$ with $0\le r<C_{\text{ipp}}$. 
As $k$ ranges over a complete residue system modulo $C_{\text{ipp}}$, 
exactly $r$ values satisfy $k+r\ge C_{\text{ipp}}$, so
\[
\sum_{k=0}^{C_{\text{ipp}}-1}\Big\lfloor\frac{k+x}{C_{\text{ipp}}}\Big\rfloor
= q\cdot C_{\text{ipp}}+r = x .
\]
Therefore the double sum equals $\sum_{x=0}^{2\varepsilon}x=
\varepsilon(2\varepsilon+1)$, and
\[
\mathbb{E}[\mathrm{DAC}]
=1+\frac{\varepsilon(2\varepsilon+1)}{(2\varepsilon+1)\,C_{\text{ipp}}}
=1+\frac{\varepsilon}{C_{\text{ipp}}} . \qedhere
\]
\end{proof}

\noindent
\textit{\textbf{Remark.}}
Lemmas~\ref{lemma:dac_aao} and~\ref{lemma:dac_obo} confirm that one-by-one fetching 
reduces $\mathbb{E}[\mathrm{DAC}]$ by $\varepsilon/C_{\text{ipp}}$ 
versus all-at-once fetching.
However, as shown in \cref{fig:leaf-fetch-strategy-comparison}, this reduction does not translate to end-to-end performance gains in practice as the dependent random I/Os issued by one-by-one fetching underutilize modern SSD concurrency.

\section{Efficient Buffer Hit Rate Estimation}\label{sec:h-estimate}
\cref{sec:cam} introduces a modular I/O cost model for disk-based learned indexes under specific cache eviction policies and page fetching strategies, requiring page reference probabilities $\Pr_{\text{req}}(i)$ as 
inputs. 
In this section, given a page $i$, we study the problem of efficiently estimating $\Pr_{\text{req}}(i)$ for point, range, and join queries on a disk-based learned index.

\subsection{Point Lookup Query}    
\label{sec:point-query}

\noindent
\textit{\textbf{Query Processing.}}
Given a query key $k$, the learned index returns an approximate position $f(k)$. 
The true position is guaranteed to lie in the error-bounded window $[f(k)-\varepsilon,f(k)+\varepsilon]$, which may span multiple pages. 
DBMS first probes the page buffer for the pages in this window and fetches any misses. 
Once the required pages are resident, it performs an in-page binary search (or other search algorithm) to locate $k$.

\noindent
\textit{\textbf{Page Reference Analysis.}}
Consider a point query workload $\mathcal{Q}$. 
For each query key $k\in\mathcal{Q}$, let $r(k)$ denote its true position in the sorted data. 
The learned index predicts $f(k)=r(k)+e$, where the error term $e$ is assumed to follow a uniform distribution $e\sim U[-\varepsilon,\varepsilon]$. 
The corresponding search window is $W(r(k),e)=[r(k)+e-\varepsilon,\; r(k)+e+\varepsilon]$.
Let page $p$ correspond to the position interval $I_p=[p\cdot C_{\text{ipp}},(p+1)\cdot C_{\text{ipp}}-1]$, where $C_{\text{ipp}}$ is the number of items per page. 
A page $p$ is accessed iff its interval intersects the last-mile window. 
Hence, for a query with true position $r$, the probability that page $p$ is accessed is
\begin{equation}\label{eq:prob_page_closed}
  \begin{aligned}
    \Pr(p\ \text{is accessed}\mid r) & = \sum_{e\in[-\varepsilon,\varepsilon]} \frac{\mathbbm{1}\bigl(W(r,e)\cap I_p\neq\phi\bigr)}{2\varepsilon+1}\\
    &=\frac{\max\bigl(0,\; U_{p,r}-L_{p,r}+1\bigr)}{2\varepsilon+1},
  \end{aligned}
\end{equation}
where $L_{p,r}=\max\bigl(-\varepsilon, p\cdot C_{\text{ipp}}-r-\varepsilon\bigr)$ and $U_{p,r}=\min\bigl(\varepsilon, (p+1)\cdot C_{\text{ipp}}-1-r+\varepsilon\bigr)$.
Aggregating over all queries of a workload $\mathcal{Q}$ yields the expected reference count of a page $p$:
\begin{equation}\label{eq:page_referred_times}
C_p=\sum\nolimits_{k\in\mathcal{Q}} \Pr(p\ \text{is accessed}\mid r(k)).
\end{equation}
We then normalize $\{C_p\}$ to obtain the page reference probabilities $\Pr_{\text{req}}(p)$. 
These probabilities, together with the data-access cost (Lemmas~\ref{lemma:dac_aao} and~\ref{lemma:dac_obo}), instantiate \cref{eq:I/O_cost_model} to estimate the I/O cost of processing $\mathcal{Q}$.
\Cref{alg:point-cam} details the physical I/O cost estimation for point queries on a learned index with error bound $\varepsilon$.

\noindent
\textit{\textbf{LUT-based Acceleration.}}
To accelerate estimation, observe that \cref{eq:prob_page_closed} depends only on the intra-page offset $s = r - q\cdot C_{\text{ipp}}$ ($q$ is the page containing the true position) and the relative page distance $d = p - q$. 
We therefore precompute a lookup table for all feasible $(d,s)$ pairs and reuse it across queries, avoiding per-position expansion over the entire data domain.
The lookup table size is $O(\varepsilon + C_{\text{ipp}})$ entries.
More precisely, since $s \in [0, C_{\text{ipp}}-1]$ and the relative page distance $d$ ranges from $-\lceil 2\varepsilon/C_{\text{ipp}} \rceil$ to $+\lceil 2\varepsilon/C_{\text{ipp}} \rceil$ (the window can reach at most $2\varepsilon$ positions left or right of the true position), the table has at most:
$C_{\text{ipp}} \cdot (2\lceil{2\varepsilon}/{C_{\text{ipp}}}\rceil + 1)\le 4\varepsilon +3 C_{\text{ipp}}$
entries. In practice this is tiny: with $\varepsilon=1024$ and $C_{\text{ipp}}=512$, fewer than $6,000$ entries, and precomputation takes milliseconds.

\begin{algorithm}[t]
\caption{CAM Estimation for Point Queries}
\label{alg:point-cam}
\DontPrintSemicolon
\small
\KwIn{
sorted keys $D$, query workload $Q$, error bound $\varepsilon$,
items per page $C_\text{ipp}$, memory budget $M$,
page size $B$, cache eviction policy $\pi$
}
\KwOut{estimated average physical I/O $\widehat{\mathit{IO}}$}

$N \gets |D|,\quad P \gets \lceil N/C_\text{ipp}\rceil$\;
$pos \gets \textsc{LocateQueries}(D,Q)$\tcp*[r]{true positions}
$\Delta \gets [-\lceil 2\varepsilon/C_{\text{ipp}} \rceil, +\lceil 2\varepsilon/C_{\text{ipp}} \rceil]$

% $\Delta \gets$ feasible relative page offsets induced by $\varepsilon$\; 
\ForEach{$d \in \Delta$}{
  \For{$s \gets 0$ \KwTo $C_\text{ipp}-1$}{
    $\mathrm{LUT}[d,s] \gets \Pr(\text{page offset }d \mid \text{in-page offset }s,\varepsilon)$\;
  }
}

$C_p \gets 0\ \text{for }\forall p \in [0,P-1]$\;

\ForEach{$r \in pos$}{
  $p \gets \lfloor r/C_\text{ipp}\rfloor,\quad s \gets r \bmod C_\text{ipp}$\;
  \ForEach{$d \in \Delta$}{
    \If{$0 \le p+d < P$}{
      $C_{p+d} \gets C_{p+d} + \mathrm{LUT}[d,s]$\;
    }
  }
}

$q_p \gets C_p / \sum_j C_j\ \text{for }\forall p$\tcp*[r]{normalization}
$M_{\mathit{idx}} \gets \textsc{IndexSize}(\varepsilon)$\tcp*[r]{est. index size}
$C \gets \left\lfloor (M-M_{\mathit{idx}})/B \right\rfloor$\tcp*[r]{buffer capacity}
$h \gets \textsc{HitRate}(\pi,C,\{q_p\})$\tcp*[r]{see \cref{subsec:hit_rate}}
$DAC \gets \textsc{ExpectedDAC}(\varepsilon,C_\text{ipp})$\tcp*[r]{see \cref{subsec:dac}}
$\widehat{\mathit{IO}} \gets (1-h)\cdot DAC$\;
\Return $\widehat{\mathit{IO}}$\;
\end{algorithm}

\noindent
\textit{\textbf{Time complexity.}}
Mapping all queries to their true positions once requires $O(|\mathcal{Q}|\log n)$ time. 
For a query with true position $r$, only pages intersecting $[r-2\varepsilon,\; r+2\varepsilon]$ may have non-zero access probability; thus the number of affected pages per query is bounded by $O(\varepsilon/C_{\text{ipp}}+1)$. 
Therefore, after preprocessing the lookup table, page reference analysis takes
$O\bigl(|\mathcal{Q}|\cdot(\varepsilon/C_{\text{ipp}}+1+\log n)\bigr)$.
In practice, the number of affected pages per query is small, often constant, making the estimation efficient enough for repeated what-if evaluation across different cache budgets and index configurations.

\noindent
\textit{\textbf{Remark.}}
As shown in \cref{alg:point-cam}, CAM estimates cache-aware I/O costs without physically constructing a specific learned index structure. 
True ranks are obtained once via standard search over the sorted in-memory data and can be reused for the same dataset--workload pair.
For highly memory-constrained scenarios, approximate ranks from lightweight sketches (e.g., histograms) suffice at minor accuracy cost.

\subsection{Range Query} \label{sec:range-query}

\noindent
\textit{\textbf{Query Processing.}}
A range query is specified by a key pair $Q=(lo,hi)$ and returns all records whose keys fall
in $[lo,hi]$. 
We first locate the predecessor positions of the two endpoints in the sorted data array. The two endpoint positions are then used to determine the data-page interval that may be touched by the range query.
The query is then processed by a single all-at-once page fetching over the page interval.

\noindent
\textit{\textbf{Page Reference Analysis.}}
For a range query $Q$ over a learned index with error $\varepsilon$, the start and end pages accessed are:
\begin{equation}\label{eq:range_conservative_pages}
\begin{aligned}
S(Q) &= \left\lfloor\frac{\max\{0,\; r(\mathit{lo})-2\varepsilon\}}
                        {C_{\text{ipp}}}\right\rfloor, \\[2pt]
E(Q) &= \left\lfloor\frac{\min\{n-1,\; r(\mathit{hi})+2\varepsilon\}}
                        {C_{\text{ipp}}}\right\rfloor .
\end{aligned}
\end{equation}
Let $\mathrm{ref}_Q(p)$ denote whether a page $p$ is referenced by a query $Q$, i.e., $\mathrm{ref}_Q(p)=\mathbbm{1}\{S(Q)\le p\le E(Q)\}$.
Aggregating over the workload $\mathcal{Q}$ yields the estimated page-reference count $C_p=\sum_{q\in\mathcal{Q}}\mathrm{ref}_Q(p)$ and the total number of logical page references 
$R=\sum_{Q\in\mathcal{Q}}\bigl(E(Q)-S(Q)+1\bigr)$. 
The average data access cost is therefore 
$\mathbb{E}[\mathrm{DAC}]=R/|\mathcal{Q}|$. 
Normalizing the reference counts gives the page request distribution $\Pr_{\mathrm{req}}(p)=C_p/R$ used by the cache model.
% In implementation, $C_p$ can be efficiently computed in one pass of data via prefix sum, avoiding iterating over every page in every query range interval.
In implementation, $C_p$ is computed by applying interval updates to a difference array and then taking a prefix sum over the page domain, avoiding iteration over every page in every query interval.

\noindent
\textit{\textbf{Time Complexity.}}
Mapping query endpoints to their true ranks costs
$O(|\mathcal Q|\log n)$. 
Constructing the page access intervals costs
$O(|\mathcal Q|)$, and the prefix sum step costs
$O(\lceil n/C_{\text{ipp}}\rceil)$.
Thus the page reference for range queries analysis runs in
$O\bigl(|\mathcal Q|\log n + |\mathcal Q| + \lceil n/C_{\text{ipp}}\rceil\bigr)$.

\subsection{Join Query}\label{subsec:join_cost}
\noindent
\textit{\textbf{Query Processing.}} We consider equi-joins $A \bowtie B$ where relations $A$ and $B$ are joined on a common attribute $c$. 
Following the index nested-loop join (INLJ) paradigm, we treat $A$ as the outer and $B$ as the inner relation, sort $A.c$, and probe each key against the learned index built on $B.c$. 
Notably, sorting operation incurs overhead but improves locality, increases buffer hit rate, and reduces random I/O. 
We further show that sorted probe order \textit{maximizes} the cache hit rate over all permutations of the same probe set.

\begin{lemma}[Optimal Probe Order]\label{cor:sorted_best}
    For any probe set $\mathcal{Q}$, sorted key order maximizes the cache hit rate among all execution permutations, provided $C \ge 1 + \lceil 2\varepsilon/C_{\text{ipp}} \rceil$. This holds regardless of the replacement policy.
\end{lemma}
\begin{proof}
Let $\sigma$ be any ordering and $\mathcal{P}$ the distinct pages it 
references with $N=|\mathcal{P}|$. Each $p\in\mathcal{P}$ incurs at least one compulsory miss, so $\text{misses}(\sigma)\ge N$. Under sorted 
order, \cref{thm:ordered_hit_rate} gives exactly one miss per page when 
$C\ge 1+\lceil 2\varepsilon/C_{\text{ipp}}\rceil$, attaining the lower 
bound. Since total references are fixed, minimizing misses maximizes hits.
\end{proof}

\noindent
\textit{\textbf{Page Reference Analysis.}}
If the outer relation is not sorted, the probe sequence degenerates into a 
general point query workload, and can be analyzed using the point query model in 
\cref{sec:point-query}.
Once the outer relation is sorted, the probe trace becomes ordered and \cref{thm:ordered_hit_rate} applies directly, yielding 
$h=(R-N)/R$, where $R$ is the total number of page references generated by the last-mile search windows and $N$ is the number of distinct pages touched. 
Unlike point and range queries, the sorted join case does not require a full page-reference probability distribution, and the cost model need only $R$ and $N$.

\noindent
\textit{\textbf{Remark.}}
The above analysis reveals two probing strategies for learned-index-based join processing: 
(1) point probing, applicable to any case, and (2) range probing, 
available when the outer relation is sorted. 
\Cref{sec:join-query} details how CAM can help partition the probe sequence and adaptively select the optimal strategy per region to minimize estimated I/O cost.
\section{CAM-based Index Tuning}    \label{sec:application}

\subsection{Methodology Overview} \label{subsec:application_overview}

Given a fixed memory budget $M$, we allocate $M_{\text{index}}$ to the 
index and assign the remainder to the page buffer, 
$M_{\text{buf}}=M-M_{\text{index}}$. CAM estimates the effective I/O 
cost as
\begin{equation}
    \text{Cost}_{\text{CAM}}(\theta; M) = \bigl(1-h(M_{\text{buf}})\bigr)\cdot \mathbb{E}[\text{DAC}(\theta)],
\end{equation}
where $\theta$ is the parameter(s) of the index to be tuned, and both 
$h(\cdot)$ and $\mathbb{E}[\text{DAC}(\cdot)]$ are derived from a representative query workload (synthetic uniform or hotspot workloads can be used when no historical workload is available). 
For any learned index structure with an analytical memory footprint and stable lookup I/O characteristics, CAM 
solves the index tuning problem by finding $\theta^\star$ such that $\text{Cost}_{\text{CAM}}(\theta; M)$ is minimized. 

% For complicated index structure without an accurate equation to pre-compute index size or the data access cost, we could obtain the I/O cost under CAM given a fixed configuration.

% \begin{figure}[t]
%     \centering
%     \includegraphics[width=\linewidth]{figures/IOs_prediction_legend_crop.pdf}
%     \vspace{1em}
%     \includegraphics[width=0.45\textwidth]{figures/books_2x2_IOs_prediction_crop.pdf}
%     \caption{Estimated vs.\ measured I/O cost for 1M lookup queries on the Books dataset under varying memory budgets $M$, where $M$ is split between the index and the page buffer.}
%     \label{fig:CAM_estimation}
% \end{figure}

\subsection{CAM-based PGM Tuning}
We instantiate CAM for PGM tuning. 
Note that the same approach generalizes to any 
learned index with error-bounded linear models, including 
PGM-Disk~\cite{DBLP:journals/pacmmod/ZhangSZ24,DBLP:journals/pvldb/FerraginaV20}, 
XIndex~\cite{tang2020xindex}, and FITing-Tree~\cite{DBLP:conf/sigmod/GalakatosMBFK19}. 
Existing analyses~\cite{DBLP:conf/icml/FerraginaLV20H} show that the 
footprint of a linear-model-based learned index on key set $\mathcal{K}$ 
scales as $M_{\text{index}}(\varepsilon) \propto |\mathcal{K}|/(2\varepsilon)$. 
For a fixed budget $M$, $\text{Cost}_{\text{CAM}}(\varepsilon; M)$ traces 
a U-shaped curve in $\varepsilon$ (\cref{fig:logical_ios_vs_epsilon}): 
increasing $\varepsilon$ initially reduces I/O by shrinking the index 
and enlarging the buffer (thereby raising $h$), but once the last-mile 
search window grows sufficiently, $\mathbb{E}[\text{DAC}(\varepsilon)]$ 
dominates and total I/O rises. The optimal error bound is therefore
\begin{equation}
\varepsilon^\star=\arg\min_{\varepsilon}\text{Cost}_{\text{CAM}}(\varepsilon; M),
\end{equation}
with buffer allocation $M_{\text{buf}}=M-M_{\text{index}}(\varepsilon^\star)$. 

However, the analytical upper bound on the PGM footprint is often too loose to accurately capture the actual index size on a given dataset.
To reduce this error without constructing an index for every candidate $\varepsilon$, we follow the fitting strategy used by the PGM tuner~\cite{DBLP:journals/pvldb/FerraginaV20}.
We first choose a small subset of representative error bounds.
Based on these samples, we fit a dataset-specific power-law model $M_{\text{index}}(\varepsilon) = a\varepsilon^{-b}+c$, 
where the parameters are initialized from a log-log regression and then refined by nonlinear least-squares fitting.
CAM then uses the fitted function to estimate the index size for all $\varepsilon$ candidates, thereby evaluating a much denser tuning space with only a small profiling cost.

% After physical index construction, we optionally replace the analytical 
% $M_{\text{index}}(\varepsilon)$ with the measured footprint and 
% re-evaluate CAM for subsequent query planning. 
% This two-stage approach keeps tuning lightweight while preserving runtime accuracy.

\subsection{CAM-based RMI Tuning}\label{subsec:rmi_tuning}
We further extend CAM to RMI-style index tuning, focusing on a two-layer RMI architecture for simplicity. 
CDFShop~\cite{DBLP:conf/sigmod/MarcusZK20} targets CPU-optimal configurations, but its cost model ignores physical I/O and buffer effects, making it unsuitable for disk-resident indexes. 
CAM extends the tuning objective to explicitly account for I/O under memory constraints.
Because RMI lacks closed-form expressions for index size and prediction error, we enumerate candidate configurations, construct each index, and estimate its cost under CAM. 
Notably, physical index construction is unavoidable, yet CAM derives expected I/O cost directly from the error bound, bypassing expensive last-mile search during evaluation and keeping per-candidate estimation lightweight. 
In our experiments, CAM-based RMI tuning outperforms CDFShop by up to \textbf{2.4\texttimes} in tuning speed and \textbf{1.7\texttimes} in query throughput (see \cref{subsec:tuning-results}).

% This reduces the multi-objective problem to a single-objective one, eliminating the need to enumerate a Pareto frontier. 

\noindent
\textit{\textbf{Data Access Cost for RMI.}}
Unlike PGM, RMI does not provide a uniform global error bound. 
Each leaf model in an RMI induces a \emph{local} error distribution. 
Consequently, the expected data access cost must be computed in a data-driven manner. 
Let $\mathcal{L}=\{1,\dots,b\}$ denote the set of leaf model indices for a fixed branching factor $b$. 
For a query key $k$, the root model routes it to 
leaf $\ell(k)\in\mathcal{L}$. 
Let $w_j = \Pr(\ell(k)=j)$ denote the probability that a query is routed to leaf $j$, which can be estimated empirically from a given workload.

For queries routed to leaf model $j$, the corresponding leaf-level error bound is $\varepsilon_j$. 
Under the page-based I/O model, the last-mile search spans a window of size proportional to the prediction error. 
Following the same procedures as in Lemmas~\ref{lemma:dac_aao} and~\ref{lemma:dac_obo}, the expected number of pages accessed for a query routed to leaf $j$ can be approximated as:
$$
\mathbb{E}[\mathrm{DAC}_j] = 1 + \frac{\lambda \varepsilon_j}{C_{\text{ipp}}}, 
$$
where $\lambda=1$ under the one-by-one strategy and $\lambda=2$ under the all-at-once strategy (see \cref{subsec:page_fetching}). 
Aggregating over all leaves, the overall expected DAC is
$$
\mathbb{E}[\mathrm{DAC}] =
\sum\nolimits_{j=1}^{b} w_j \cdot \mathbb{E}[\mathrm{DAC}_j].
$$
Thus, unlike fixed-error learned indexes, RMI's DAC depends jointly on the index configuration and the workload-induced routing distribution.

\noindent
\textit{\textbf{Page Reference Analysis for RMI.}}
The analysis in \cref{sec:point-query} assumes a single global error bound 
$\varepsilon$. 
For a typical RMI structure, the last-mile search window is instead determined by the routed leaf model. 
Let $\ell(k)$ denote the leaf selected for query $k$, with local error bound $\varepsilon_{\ell(k)}$. 
The expected reference count of page $p$ is:
\[
C_p^{\mathrm{RMI}}=\sum\nolimits_{k\in\mathcal{Q}}\Pr\bigl(p\ \text{is accessed}\mid r(k),\varepsilon_{\ell(k)}\bigr),
\]
where the conditional probability is evaluated via 
\cref{eq:prob_page_closed} with the leaf-specific bound $\varepsilon_{\ell(k)}$ in place of the global $\varepsilon$. 
Thus, unlike the fixed-$\varepsilon$ case, the RMI page reference distribution is a workload-weighted mixture of leaf-specific access patterns.

\begin{figure}[t]
    \centering
    \includegraphics[width=0.85\linewidth]{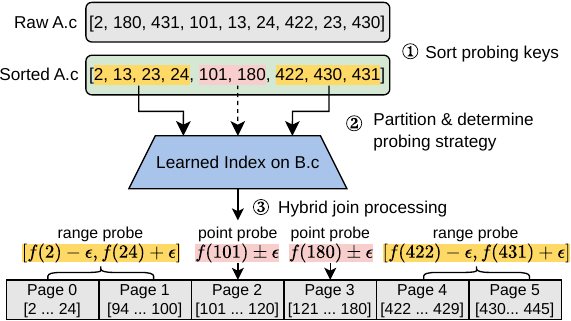}
    \caption{Overview of the hybrid join strategy. Outer-relation probe keys are sorted, partitioned into density-aware regions, and processed via range probes (dense regions) or point probes (sparse regions) according to CAM's I/O cost estimation.}
    \label{fig:join_query_workflow}
\end{figure}

\section{CAM-based Hybrid Join Strategy}\label{sec:join-query}

\subsection{Methodology Overview}
Prior discussion in \cref{subsec:join_cost} and existing work~\cite{DBLP:conf/acda/Chesetti025} 
identify two probing strategies for learned-index joins: point probing, applicable to any workload, and range probing, enabled by sorting the outer relation.
We sort the probe relation by default as our experiments confirm that the resulting locality and buffer-hit gains outweigh the sorting overhead. 
Building on this insight, we consider the range probe, which groups nearby probe keys (from sorted $A.c$) into range probes over $B.c$ and filters false positives (keys not in $A.c$). 
Range probes amortize per-key traversal and OS overheads (e.g., cache operations and syscalls) but may read redundant pages in sparse regions. 
Accordingly, as illustrated in \cref{fig:join_query_workflow}, we adopt a 
\emph{hybrid} strategy for learned-index-based join processing: 
it partitions the sorted probe keys $\mathcal{Q}$ into segments and applies CAM to select point probing or range probing with the objective of minimizing the total I/O cost.

\subsection{Cost Modeling} 
Join processing executes a mixture of point and range probes. However, 
unlike the standalone workloads analyzed earlier, sorted join probing 
produces a \emph{sorted workload} that violates the IRM assumption 
underlying Che's approximation. By \cref{thm:ordered_hit_rate}, the 
buffer hit rate in this regime reduces to the closed-form $(R-N)/R$ 
and is no longer policy-dependent. 
Moreover, join execution is not purely I/O-bound: 
the CPU overhead of per-key index traversal and range filtering is usually non-negligible. 
We therefore estimate join cost via a lightweight model that accounts for both CPU work and physical page misses.

For a segment $S\subseteq\mathcal{Q}$, let $N_S$ denote the number of probe keys, $d_S$ the 
number of distinct pages touched by point probing, and $K_S$ the page 
span of the corresponding range probe. The point and range probe costs are
\begin{equation}\label{eq:cost_for_join}
\begin{aligned}
    \text{Cost}_{\text{point}}(S) &= \delta + \alpha N_S + \lambda_{\text{point}}\, d_S, \\
    \text{Cost}_{\text{range}}(S) &= \eta + (\beta + \lambda_{\text{range}})\, K_S,
\end{aligned}
\end{equation}
where $\alpha$ captures per-key CPU work (traversal, last-mile search, cache updates), 
$\beta$ is the per-page CPU cost of scanning and filtering, 
and $\lambda_{\text{point}}$/$\lambda_{\text{range}}$ denote the average latency per physical page miss. 
The intercepts $\delta$ and $\eta$ absorb residual effects such as measurement bias. 
We fit all 
parameters via short calibration runs (see \cref{tab:workload_mixture} 
and \cref{sec:exp-hybrid-join}).

\begin{algorithm}[t]
\caption{Greedy Join Probe Partitioning}
\label{alg:join-partition}
\small
\DontPrintSemicolon
\KwIn{sorted probe keys $\mathcal{Q}$, items per page $C_{\text{ipp}}$, minimum segment size $N_{\min}$, maximum page span $K_{\max}$, margin $\gamma$, cost parameters $\alpha,\beta,\eta,\delta,\lambda_{\text{point}},\lambda_{\text{range}}$}
\KwOut{partitioned segments $\mathcal{S}_p$, probe strategy bitmap $B$}

$\{[l_q,r_q]\}_{q\in \mathcal{Q}} \gets \textsc{PageIntervals}(\mathcal{Q},\varepsilon,C_{\text{ipp}})$\;
$\mathcal{S}_p \gets \emptyset$,\quad $B \gets [\ ]$,\quad $i \gets 0$\;

\While{$i < |\mathcal{Q}|$}{
  $j \gets i$,\quad $W \gets \emptyset$\;
  
  \While{$j < |\mathcal{Q}|$}{
    $W \gets W \cup \{q_j\}$ with page access interval $[l_j,r_j]$\;
    $N \gets |W|$\;
    $K \gets \text{page span under range probe}$\;
    $d \gets \text{distinct pages under point probes}$\;
    
    \If{$N \ge N_{\min}$}{
      $\mathrm{Cost}_{p} \gets \delta + \alpha N + \lambda_{\text{point}}\, d$\;
      $\mathrm{Cost}_{r} \gets \eta + (\beta + \lambda_{\text{range}})\, K$\;
      
      \If{$K \ge K_{\max} \lor \mathrm{Cost}_{r} \le (1-\gamma)\cdot\mathrm{Cost}_{p}$}{
        \textbf{break}\;
      }
    }
    $j \gets j + 1$\;
  }
  $S \gets W$\;
  recompute $\mathrm{Cost}_{p}$ and $\mathrm{Cost}_{r}$ on $S$\;
  $b \gets \bigl(|S| \ge N_{\min}\bigr) \land \bigl(\mathrm{Cost}_{r} \le (1-\gamma)\cdot \mathrm{Cost}_{p}\bigr)$\;
  
  $\mathcal{S}_p \gets \mathcal{S}_p \cup \{S\}$\;
  append $b$ to $B$\tcp*[r]{0: point 1: range}
  $i \gets j + 1$\;
}
\Return $\mathcal{S}_p, B$\;
\end{algorithm}
% where $\alpha$ captures per-key CPU work (index traversal, last-mile checks, cache updates), $\delta$ and $\eta$ is the , $\beta$ is the per-page CPU cost of scanning and filtering, and $\lambda_{\text{point}}$/$\lambda_{\text{range}}$ represent the average latency per physical page miss under point/range execution. 
% We fit these hyperparameters via short calibration runs (see \cref{tab:workload_mixture} and \cref{sec:exp-hybrid-join}).

% \noindent
% \textit{\textbf{Partitioning Details.}} 
\subsection{Partitioning Algorithm Details}
Based on \cref{eq:cost_for_join}, we partition the sorted probe stream via a greedy single-pass strategy. 
As shown in \cref{alg:join-partition}, we incrementally extend the current segment while tracking $(N_S, d_S, K_S)$ 
and the corresponding point/range cost estimates. A segment is closed when its page span reaches $K_{\max}=8192$, or range probing outperforms point probing by margin $\gamma$ after the segment has accumulated at least $N_{\min}=1024$ probes, which prevents over-fragmentation. 
Segment boundaries and probing modes are stored compactly as an array of lengths paired with a bitmask.

\section{Experiments} \label{sec:exp}
\subsection{Experimental Setup}  \label{sec:exp-setup}

\noindent
\textit{\textbf{Environment.}}
All experiments run on a dual-socket Ubuntu 24.04 server with two Intel Xeon Gold 6430 CPUs, 512\,GiB DRAM, and a 4\,TB NVMe SSD that stores all datasets and index files. 

\noindent
\textit{\textbf{Datasets and Baselines.}}
Following prior learned index studies~\cite{DBLP:journals/pvldb/FerraginaV20,DBLP:conf/icml/FerraginaLV20H,DBLP:journals/pvldb/SunZL23,tang2020xindex,DBLP:conf/sigmod/KipfMRSKK020,DBLP:conf/sigmod/KraskaBCDP18}, we adopt four real datasets from a recent benchmark~\cite{DBLP:journals/pvldb/MarcusKRSMK0K20}: {books}, {fb}, {osm}, and {wiki}. 
Each dataset contains 200\,M sorted 64-bit uint64 keys. 
Dataset statistics are omitted due to space constraints and are available 
in~\cite{DBLP:journals/pvldb/MarcusKRSMK0K20}.
For baselines, we implement and compare: (1) CAM-$x$: our cache-aware I/O cost model that derives page reference probabilities from an $x$\% workload sample under LRU (the default eviction strategy), 
(2) Replay-$x$: 
trace-driven simulation that replays $x$\% of the query trace to obtain the exact hit rate, 
and (3) LPM: logical page model that directly counts all logical pages touched by the last-mile search.
Our implementation is publicly available at \url{https://github.com/collectcrop/CAM}.

\noindent
\textit{\textbf{Workloads.}}
For point queries and join probing, we generate keys from a three-component mixture (\cref{tab:workload_mixture}): (1) hot-spot regions (small contiguous ranges with high skewness), (2) a Zipf distribution over the full key domain, 
and (3) a residual uniform component. 
We default to workload w4, which simultaneously exhibits strong locality and long-tail behavior. 
For range queries, we sample lower-bound keys from the same generator and 
pair them with random range lengths.

\noindent
\textit{\textbf{Metric of Interest.}} 
We report different metrics according to the goal of each experiment. 
(1) For I/O cost modeling evaluation, the primary metrics are \emph{Q-error} and \emph{estimation time}. 
The (average) Q-error measures how closely the average I/O cost predicted by CAM matches the actual physical I/O cost which is observed by replaying workload $\mathcal{Q}$, and is computed as:
\[
\text{Q-error}(\mathcal{Q}) = 
\frac{1}{|\mathcal{Q}|}\sum\nolimits_{Q\in\mathcal{Q}}\max\Big({\mathrm{IO}_Q}/{\widehat{\mathrm{IO}}_Q}, {\widehat{\mathrm{IO}}_Q}/{\mathrm{IO}_Q}\Big),
\]
where $\widehat{\mathrm{IO}}_Q$ is the CAM-estimated I/O and $\mathrm{IO}_Q$ is the measured physical I/O for query $Q$.  
% Estimation time captures the overhead of running CAM itself, which determines whether CAM can be used inside a practical knob tuning loop.
(2) For index tuning experiments, the primary metric is query throughput (queries per second, QPS), which directly reflects the quality of the configuration selected by each tuner under a fixed memory budget. 
We also report tuning time to quantify the optimization overhead of CAM-guided tuning versus existing methods such as CDFShop~\cite{DBLP:conf/sigmod/MarcusZK20}. 
(3) For join experiments, we report end-to-end join processing time, capturing the full pipeline of sorting, probing, cache management, and result filtering. 

\begin{table}[t]
  \centering
  \caption{Cost model parameter fitting results (see \cref{sec:exp-hybrid-join} for more details) and the mixture proportions of the query distributions used to generate the workload.}
  \label{tab:workload_mixture}
  \scriptsize
  \begin{tabular}{c|cccc}
    \toprule
    \textbf{Parameters} & \textbf{Workload} & \textbf{hotspot} & \textbf{zipf} & \textbf{uniform}\\
    \midrule
    $\lambda_{\text{point}}=1.19{\times}10^{-6}$ & w1 & 0\%   & 0\%   & 100\%  \\
    $\lambda_{\text{range}}=4.66{\times}10^{-7}$ & w2 & 0\%   & 100\%   & 0\% \\
    $\alpha=1.64{\times}10^{-6}$ & w3 & 100\% & 0\% & 0\% \\
    $\beta=1.72{\times}10^{-6}$ & w4 & 40\% & 30\% & 30\%  \\
    $\eta=4.42{\times}10^{-6}$ & w5 & 20\% & 20\% & 60\% \\
    $\delta=5.00{\times}10^{-3}$ & w6 & 10\% & 10\% & 80\% \\ 
    \bottomrule
  \end{tabular}
\end{table}
\subsection{Evaluation of Accuracy and Efficiency of CAM}
To evaluate CAM's accuracy in estimating physical I/Os, we instantiate a disk-based PGM-Index~\cite{DBLP:journals/pacmmod/ZhangSZ24} backed by a 128\,MB LRU page buffer. 
\Cref{tab:cam-replay-lpm-qerror,tab:range-cam-replay-lpm-qerror} 
report the estimation q-error and estimation time for CAM, Replay, and 
LPM on 1\,M point and range queries, respectively.
The reported results are averaged across 9 different error bound configurations.

% This setting examines whether CAM can infer cache behavior and I/O cost from limited workload samples, rather than relying on full workload replay.

% \begin{figure*}[t]
%     \centering
%     \includegraphics[width=0.4\textwidth]{figures/exp/point_io_workload_legend.pdf}
%     \includegraphics[width=0.9\textwidth]{figures/exp/point_io_estimation_accuracy.pdf}
%     \includegraphics[width=0.9\textwidth]{figures/exp/point_io_estimation_time.pdf}
%     \caption{Accuracy and Efficiency of CAM I/O estimation for PGM under point workload.}
%     \label{fig:acc-eff-of-CAM}
% \end{figure*}
% \begin{figure*}[t]
%     \centering
%     \includegraphics[width=0.4\textwidth]{figures/exp/range_io_workload_legend.pdf}
%     \includegraphics[width=0.9\textwidth]{figures/exp/range_io_estimation_accuracy.pdf}
%     \includegraphics[width=0.9\textwidth]{figures/exp/range_io_estimation_time.pdf}
%     \caption{Accuracy and Efficiency of CAM I/O estimation for PGM under range workload.}
%     \label{fig:acc-eff-of-CAM}
% \end{figure*}

% Required in the preamble:
% \usepackage[table]{xcolor}

\begin{table*}[t] \centering 
\definecolor{camgreen}{RGB}{205,230,209}
\definecolor{replayred}{RGB}{253,222,227} 
\definecolor{lpmgray}{RGB}{238,238,238} 
\newcommand{\CAM}{\cellcolor{camgreen}} \newcommand{\REP}{\cellcolor{replayred}} \newcommand{\LPM}{\cellcolor{lpmgray}} \caption{Comparison of CAM, Replay, and LPM across point query workloads and sample rates. Parenthesized values denote speedup over Replay-100 under identical settings. Q-Err of $1.00$ indicates perfect accuracy.} \label{tab:cam-replay-lpm-qerror} \scriptsize \resizebox{\textwidth}{!}{ \begin{tabular}{llcccccccc} \toprule \multirow{2}{*}{\textbf{Dataset}} & \multirow{2}{*}{\textbf{Method}} & \multicolumn{2}{c}{\textbf{w1}} & \multicolumn{2}{c}{\textbf{w2}} & \multicolumn{2}{c}{\textbf{w4}} & \multicolumn{2}{c}{\textbf{w6}} \\ \cmidrule(lr){3-4} \cmidrule(lr){5-6} \cmidrule(lr){7-8} \cmidrule(lr){9-10} & & \textbf{Total Time(s)} & \textbf{Mean Q-Err} & \textbf{Total Time(s)} & \textbf{Mean Q-Err} & \textbf{Total Time(s)} & \textbf{Mean Q-Err} & \textbf{Total Time(s)} & \textbf{Mean Q-Err} \\ \midrule \multirow{9}{*}{books} & \CAM CAM-10 & \CAM 0.41 (48.59$\times$) & \CAM 1.106 & \CAM 0.13 (105.69$\times$) & \CAM 4.071 & \CAM 0.21 (75.00$\times$) & \CAM 1.491 & \CAM 0.35 (51.26$\times$) & \CAM 1.122 \\ & \CAM CAM-30 & \CAM 1.28 (15.56$\times$) & \CAM 1.033 & \CAM 0.26 (52.85$\times$) & \CAM 1.777 & \CAM 0.57 (27.63$\times$) & \CAM 1.084 & \CAM 1.06 (16.92$\times$) & \CAM 1.029 \\ & \CAM CAM-50 & \CAM 2.03 (9.81$\times$) & \CAM 1.020 & \CAM 0.40 (34.35$\times$) & \CAM 1.340 & \CAM 0.92 (17.12$\times$) & \CAM 1.041 & \CAM 1.68 (10.68$\times$) & \CAM 1.014 \\ & \CAM CAM-100 & \CAM 3.75 (5.31$\times$) & \CAM 1.011 & \CAM 0.71 (19.35$\times$) & \CAM 1.072 & \CAM 1.67 (9.43$\times$) & \CAM 1.016 & \CAM 3.12 (5.75$\times$) & \CAM 1.004 \\ & \REP Replay-10 & \REP 12.27 (1.62$\times$) & \REP 1.002 & \REP 10.86 (1.27$\times$) & \REP 1.405 & \REP 10.96 (1.44$\times$) & \REP 1.014 & \REP 11.09 (1.62$\times$) & \REP 1.003 \\ & \REP Replay-30 & \REP 13.66 (1.46$\times$) & \REP 1.001 & \REP 11.31 (1.21$\times$) & \REP 1.165 & \REP 12.29 (1.28$\times$) & \REP 1.005 & \REP 12.62 (1.42$\times$) & \REP 1.002 \\ & \REP Replay-50 & \REP 15.58 (1.28$\times$) & \REP 1.001 & \REP 12.04 (1.14$\times$) & \REP 1.082 & \REP 13.03 (1.21$\times$) & \REP 1.002 & \REP 14.11 (1.27$\times$) & \REP 1.001 \\ & \REP Replay-100 & \REP 19.92 (1.00$\times$) & \REP 1.000 & \REP 13.74 (1.00$\times$) & \REP 1.000 & \REP 15.75 (1.00$\times$) & \REP 1.000 & \REP 17.94 (1.00$\times$) & \REP 1.000 \\ & \LPM LPM & \LPM 14.50 (1.37$\times$) & \LPM 1.027 & \LPM 12.67 (1.08$\times$) & \LPM 21.984 & \LPM 13.86 (1.14$\times$) & \LPM 2.598 & \LPM 14.54 (1.23$\times$) & \LPM 1.197 \\ \midrule \multirow{9}{*}{fb} & \CAM CAM-10 & \CAM 0.40 (59.40$\times$) & \CAM 1.105 & \CAM 0.09 (177.67$\times$) & \CAM 3.974 & \CAM 0.18 (104.83$\times$) & \CAM 1.514 & \CAM 0.33 (64.61$\times$) & \CAM 1.125 \\ & \CAM CAM-30 & \CAM 1.19 (19.97$\times$) & \CAM 1.033 & \CAM 0.18 (88.83$\times$) & \CAM 1.772 & \CAM 0.52 (36.29$\times$) & \CAM 1.097 & \CAM 0.99 (21.54$\times$) & \CAM 1.031 \\ & \CAM CAM-50 & \CAM 1.90 (12.51$\times$) & \CAM 1.020 & \CAM 0.36 (44.42$\times$) & \CAM 1.341 & \CAM 0.86 (21.94$\times$) & \CAM 1.049 & \CAM 1.60 (13.32$\times$) & \CAM 1.015 \\ & \CAM CAM-100 & \CAM 3.73 (6.37$\times$) & \CAM 1.011 & \CAM 0.67 (23.87$\times$) & \CAM 1.071 & \CAM 1.65 (11.44$\times$) & \CAM 1.020 & \CAM 3.17 (6.73$\times$) & \CAM 1.004 \\ & \REP Replay-10 & \REP 14.25 (1.67$\times$) & \REP 1.001 & \REP 12.82 (1.25$\times$) & \REP 1.436 & \REP 13.40 (1.41$\times$) & \REP 1.012 & \REP 13.49 (1.58$\times$) & \REP 1.002 \\ & \REP Replay-30 & \REP 22.57 (1.05$\times$) & \REP 1.000 & \REP 13.48 (1.19$\times$) & \REP 1.166 & \REP 14.56 (1.30$\times$) & \REP 1.002 & \REP 14.76 (1.44$\times$) & \REP 1.001 \\ & \REP Replay-50 & \REP 18.32 (1.30$\times$) & \REP 1.000 & \REP 14.22 (1.12$\times$) & \REP 1.080 & \REP 16.03 (1.18$\times$) & \REP 1.002 & \REP 16.58 (1.29$\times$) & \REP 1.001 \\ & \REP Replay-100 & \REP 23.76 (1.00$\times$) & \REP 1.000 & \REP 15.99 (1.00$\times$) & \REP 1.000 & \REP 18.87 (1.00$\times$) & \REP 1.000 & \REP 21.32 (1.00$\times$) & \REP 1.000 \\ & \LPM LPM & \LPM 18.78 (1.27$\times$) & \LPM 1.027 & \LPM 15.61 (1.02$\times$) & \LPM 22.018 & \LPM 17.48 (1.08$\times$) & \LPM 2.653 & \LPM 21.60 (0.99$\times$) & \LPM 1.199 \\ \midrule \multirow{9}{*}{osm} & \CAM CAM-10 & \CAM 0.40 (57.38$\times$) & \CAM 1.105 & \CAM 0.09 (173.44$\times$) & \CAM 4.110 & \CAM 0.18 (102.72$\times$) & \CAM 1.536 & \CAM 0.34 (60.00$\times$) & \CAM 1.143 \\ & \CAM CAM-30 & \CAM 1.11 (20.68$\times$) & \CAM 1.033 & \CAM 0.18 (86.72$\times$) & \CAM 1.780 & \CAM 0.47 (39.34$\times$) & \CAM 1.096 & \CAM 0.94 (21.70$\times$) & \CAM 1.048 \\ & \CAM CAM-50 & \CAM 1.81 (12.68$\times$) & \CAM 1.020 & \CAM 0.34 (45.91$\times$) & \CAM 1.341 & \CAM 0.76 (24.33$\times$) & \CAM 1.049 & \CAM 1.53 (13.33$\times$) & \CAM 1.031 \\ & \CAM CAM-100 & \CAM 3.57 (6.43$\times$) & \CAM 1.011 & \CAM 0.65 (24.02$\times$) & \CAM 1.071 & \CAM 1.50 (12.33$\times$) & \CAM 1.020 & \CAM 3.05 (6.69$\times$) & \CAM 1.019 \\ & \REP Replay-10 & \REP 13.46 (1.71$\times$) & \REP 1.002 & \REP 11.82 (1.32$\times$) & \REP 1.395 & \REP 12.36 (1.50$\times$) & \REP 1.008 & \REP 12.33 (1.65$\times$) & \REP 1.002 \\ & \REP Replay-30 & \REP 15.84 (1.45$\times$) & \REP 1.000 & \REP 12.61 (1.24$\times$) & \REP 1.164 & \REP 13.77 (1.34$\times$) & \REP 1.002 & \REP 14.14 (1.44$\times$) & \REP 1.001 \\ & \REP Replay-50 & \REP 17.76 (1.29$\times$) & \REP 1.000 & \REP 13.48 (1.16$\times$) & \REP 1.082 & \REP 15.03 (1.23$\times$) & \REP 1.001 & \REP 16.20 (1.26$\times$) & \REP 1.001 \\ & \REP Replay-100 & \REP 22.95 (1.00$\times$) & \REP 1.000 & \REP 15.61 (1.00$\times$) & \REP 1.000 & \REP 18.49 (1.00$\times$) & \REP 1.000 & \REP 20.40 (1.00$\times$) & \REP 1.000 \\ & \LPM LPM & \LPM 17.70 (1.30$\times$) & \LPM 1.027 & \LPM 14.83 (1.05$\times$) & \LPM 22.019 & \LPM 17.79 (1.04$\times$) & \LPM 2.579 & \LPM 17.50 (1.17$\times$) & \LPM 1.206 \\ \midrule \multirow{9}{*}{wiki} & \CAM CAM-10 & \CAM 0.39 (45.18$\times$) & \CAM 1.105 & \CAM 0.08 (141.38$\times$) & \CAM 4.120 & \CAM 0.15 (88.93$\times$) & \CAM 1.575 & \CAM 0.25 (61.56$\times$) & \CAM 1.131 \\ & \CAM CAM-30 & \CAM 1.26 (13.98$\times$) & \CAM 1.032 & \CAM 0.22 (51.41$\times$) & \CAM 1.767 & \CAM 0.38 (35.11$\times$) & \CAM 1.144 & \CAM 0.69 (22.30$\times$) & \CAM 1.039 \\ & \CAM CAM-50 & \CAM 2.41 (7.31$\times$) & \CAM 1.020 & \CAM 0.34 (33.26$\times$) & \CAM 1.345 & \CAM 0.61 (21.87$\times$) & \CAM 1.096 & \CAM 1.09 (14.12$\times$) & \CAM 1.024 \\ & \CAM CAM-100 & \CAM 3.63 (4.85$\times$) & \CAM 1.011 & \CAM 0.66 (17.14$\times$) & \CAM 1.071 & \CAM 1.21 (11.02$\times$) & \CAM 1.066 & \CAM 2.17 (7.09$\times$) & \CAM 1.013 \\ & \REP Replay-10 & \REP 8.97 (1.96$\times$) & \REP 1.002 & \REP 8.18 (1.38$\times$) & \REP 1.391 & \REP 8.24 (1.62$\times$) & \REP 1.015 & \REP 8.53 (1.80$\times$) & \REP 1.004 \\ & \REP Replay-30 & \REP 10.99 (1.60$\times$) & \REP 1.001 & \REP 8.97 (1.26$\times$) & \REP 1.169 & \REP 9.40 (1.42$\times$) & \REP 1.005 & \REP 10.10 (1.52$\times$) & \REP 1.001 \\ & \REP Replay-50 & \REP 12.72 (1.39$\times$) & \REP 1.000 & \REP 9.58 (1.18$\times$) & \REP 1.079 & \REP 10.58 (1.26$\times$) & \REP 1.002 & \REP 11.55 (1.33$\times$) & \REP 1.001 \\ & \REP Replay-100 & \REP 17.62 (1.00$\times$) & \REP 1.000 & \REP 11.31 (1.00$\times$) & \REP 1.000 & \REP 13.34 (1.00$\times$) & \REP 1.000 & \REP 15.39 (1.00$\times$) & \REP 1.000 \\ & \LPM LPM & \LPM 11.55 (1.53$\times$) & \LPM 1.027 & \LPM 10.14 (1.12$\times$) & \LPM 22.050 & \LPM 11.04 (1.21$\times$) & \LPM 2.817 & \LPM 11.67 (1.32$\times$) & \LPM 1.212 \\ \bottomrule \end{tabular} }
\vspace{-3ex}
\end{table*}

\noindent
\textit{\textbf{Point Workload Experiments.}}
As shown in \cref{tab:cam-replay-lpm-qerror}, CAM achieves a favorable 
accuracy-efficiency trade-off. 
Across all datasets, workloads, and sample rates, CAM is on average \textbf{29.3$\times$} faster than Replay at the same sample rate, while keeping the median Q-error within \textbf{1.071$\times$} of the ground truth. 
% CAM's accuracy improves with the sample rate.
% For example, on w2 over books, the average Q-error drops from $4.071$ at $10\%$ sampling to $1.072$ at $100\%$ sampling.
% This trend shows that CAM converts workload samples into increasingly accurate cache-behavior estimates, while its low median error indicates that even partial samples are often sufficient to capture the dominant locality patterns.
As a weak baseline, LPM (counting logical page reads) incurs the highest Q-error in most settings because it ignores cache effects entirely and cannot capture the interplay among workload locality, page reuse, and buffer replacement. 

The acceleration of CAM over replay-based estimation mainly comes from avoiding physical index construction for every candidate configuration. 
Replay-based estimation must build the corresponding index configuration and execute the sampled queries to reproduce cache behavior, which introduces significant overhead when multiple $\varepsilon$ values need to be evaluated. 
In contrast, CAM analytically derives page-reference probabilities and cache hit ratios. 
In addition, CAM precomputes the query histogram once for a given dataset and workload. Since this histogram is independent of $\varepsilon$, it can be reused across different index configurations, making CAM particularly scalable in multi-$\varepsilon$ estimation scenarios.

Under the highly skewed workload w2, both CAM and Replay exhibit higher Q-error with small samples because random prefixes fail to capture the full distribution; the error rapidly converges to 1 as sampling increases. 
As shown in \cref{tab:workload_mixture}, w2's Zipf pattern touches a small working set, reducing the distinct pages CAM must track and lowering its estimation cost. 
Consequently, CAM achieves even shorter estimation times on skewed workloads.

\noindent
\textit{\textbf{Range Workload Experiments.}}
\cref{tab:range-cam-replay-lpm-qerror} reports the results for range workloads. 
Compared with point queries, range queries introduce more complex page-reference patterns because each query may access multiple consecutive pages and the boundary pages depend jointly on the range endpoints, prediction error, and data distribution. This makes range-workload estimation more expensive than point-workload estimation.
Despite this additional complexity, CAM still achieves accurate and efficient estimation. 
Across all datasets, workloads, and sample rates, CAM is on average \textbf{25.2$\times$} faster than replay-based estimation at the same sample rate, while keeping the median Q-error close to \textbf{1.030}.

\begin{table*}[t] \centering 
\definecolor{camgreen}{RGB}{205,230,209}
\definecolor{replayred}{RGB}{253,222,227}  
\definecolor{lpmgray}{RGB}{238,238,238} \newcommand{\CAM}{\cellcolor{camgreen}} \newcommand{\REP}{\cellcolor{replayred}} \newcommand{\LPM}{\cellcolor{lpmgray}} \caption{Comparison of CAM, Replay, and LPM across range query workloads and sample rates. Parenthesized values denote speedup over Replay-100 under identical settings. Q-Err of $1.00$ indicates perfect accuracy.} \label{tab:range-cam-replay-lpm-qerror} \scriptsize \resizebox{\textwidth}{!}{ \begin{tabular}{llcccccccc} \toprule \multirow{2}{*}{\textbf{Dataset}} & \multirow{2}{*}{\textbf{Method}} & \multicolumn{2}{c}{\textbf{w1}} & \multicolumn{2}{c}{\textbf{w2}} & \multicolumn{2}{c}{\textbf{w4}} & \multicolumn{2}{c}{\textbf{w6}} \\ \cmidrule(lr){3-4} \cmidrule(lr){5-6} \cmidrule(lr){7-8} \cmidrule(lr){9-10} & & \textbf{Total Time(s)} & \textbf{Mean Q-Err} & \textbf{Total Time(s)} & \textbf{Mean Q-Err} & \textbf{Total Time(s)} & \textbf{Mean Q-Err} & \textbf{Total Time(s)} & \textbf{Mean Q-Err} \\ \midrule \multirow{9}{*}{books} & \CAM CAM-10 & \CAM 0.58 (48.10$\times$) & \CAM 1.021 & \CAM 0.15 (113.00$\times$) & \CAM 4.057 & \CAM 0.45 (47.56$\times$) & \CAM 1.151 & \CAM 0.74 (33.72$\times$) & \CAM 1.032 \\ & \CAM CAM-30 & \CAM 1.34 (20.82$\times$) & \CAM 1.023 & \CAM 0.26 (65.19$\times$) & \CAM 1.753 & \CAM 0.95 (22.53$\times$) & \CAM 1.022 & \CAM 1.80 (13.86$\times$) & \CAM 1.021 \\ & \CAM CAM-50 & \CAM 2.18 (12.80$\times$) & \CAM 1.030 & \CAM 0.42 (40.36$\times$) & \CAM 1.358 & \CAM 1.52 (14.08$\times$) & \CAM 1.012 & \CAM 2.94 (8.49$\times$) & \CAM 1.029 \\ & \CAM CAM-100 & \CAM 4.21 (6.63$\times$) & \CAM 1.035 & \CAM 0.66 (25.68$\times$) & \CAM 1.078 & \CAM 2.93 (7.30$\times$) & \CAM 1.024 & \CAM 7.12 (3.50$\times$) & \CAM 1.033 \\ & \REP Replay-10 & \REP 12.83 (2.17$\times$) & \REP 1.001 & \REP 11.18 (1.52$\times$) & \REP 1.314 & \REP 16.08 (1.33$\times$) & \REP 1.005 & \REP 12.12 (2.06$\times$) & \REP 1.002 \\ & \REP Replay-30 & \REP 15.71 (1.78$\times$) & \REP 1.000 & \REP 12.56 (1.35$\times$) & \REP 1.107 & \REP 13.70 (1.56$\times$) & \REP 1.002 & \REP 15.04 (1.66$\times$) & \REP 1.002 \\ & \REP Replay-50 & \REP 19.37 (1.44$\times$) & \REP 1.000 & \REP 13.67 (1.24$\times$) & \REP 1.049 & \REP 15.79 (1.36$\times$) & \REP 1.001 & \REP 17.94 (1.39$\times$) & \REP 1.001 \\ & \REP Replay-100 & \REP 27.90 (1.00$\times$) & \REP 1.000 & \REP 16.95 (1.00$\times$) & \REP 1.000 & \REP 21.40 (1.00$\times$) & \REP 1.000 & \REP 24.95 (1.00$\times$) & \REP 1.000 \\ & \LPM LPM & \LPM 17.38 (1.61$\times$) & \LPM 1.027 & \LPM 14.40 (1.18$\times$) & \LPM 20.575 & \LPM 19.82 (1.08$\times$) & \LPM 2.570 & \LPM 17.25 (1.45$\times$) & \LPM 1.191 \\ \midrule \multirow{9}{*}{fb} & \CAM CAM-10 & \CAM 0.54 (59.00$\times$) & \CAM 1.022 & \CAM 0.11 (181.82$\times$) & \CAM 3.989 & \CAM 0.41 (61.15$\times$) & \CAM 1.153 & \CAM 0.97 (30.51$\times$) & \CAM 1.035 \\ & \CAM CAM-30 & \CAM 1.31 (24.32$\times$) & \CAM 1.023 & \CAM 0.19 (105.26$\times$) & \CAM 1.767 & \CAM 0.89 (28.17$\times$) & \CAM 1.026 & \CAM 2.29 (12.92$\times$) & \CAM 1.017 \\ & \CAM CAM-50 & \CAM 2.12 (15.03$\times$) & \CAM 1.030 & \CAM 0.28 (71.43$\times$) & \CAM 1.369 & \CAM 1.46 (17.17$\times$) & \CAM 1.011 & \CAM 2.93 (10.10$\times$) & \CAM 1.026 \\ & \CAM CAM-100 & \CAM 4.12 (7.73$\times$) & \CAM 1.035 & \CAM 0.62 (32.26$\times$) & \CAM 1.090 & \CAM 2.85 (8.80$\times$) & \CAM 1.019 & \CAM 5.70 (5.19$\times$) & \CAM 1.032 \\ & \REP Replay-10 & \REP 15.04 (2.12$\times$) & \REP 1.001 & \REP 13.52 (1.48$\times$) & \REP 1.349 & \REP 14.25 (1.76$\times$) & \REP 1.008 & \REP 14.21 (2.08$\times$) & \REP 1.002 \\ & \REP Replay-30 & \REP 19.43 (1.64$\times$) & \REP 1.000 & \REP 14.85 (1.35$\times$) & \REP 1.111 & \REP 16.66 (1.50$\times$) & \REP 1.001 & \REP 18.24 (1.62$\times$) & \REP 1.001 \\ & \REP Replay-50 & \REP 22.59 (1.41$\times$) & \REP 1.000 & \REP 22.36 (0.89$\times$) & \REP 1.052 & \REP 19.06 (1.32$\times$) & \REP 1.001 & \REP 21.65 (1.37$\times$) & \REP 1.001 \\ & \REP Replay-100 & \REP 31.86 (1.00$\times$) & \REP 1.000 & \REP 20.00 (1.00$\times$) & \REP 1.000 & \REP 25.07 (1.00$\times$) & \REP 1.000 & \REP 29.59 (1.00$\times$) & \REP 1.000 \\ & \LPM LPM & \LPM 22.52 (1.41$\times$) & \LPM 1.027 & \LPM 19.10 (1.05$\times$) & \LPM 20.661 & \LPM 20.10 (1.25$\times$) & \LPM 2.536 & \LPM 21.51 (1.38$\times$) & \LPM 1.191 \\ \midrule \multirow{9}{*}{osm} & \CAM CAM-10 & \CAM 0.55 (58.71$\times$) & \CAM 1.021 & \CAM 0.10 (198.80$\times$) & \CAM 4.059 & \CAM 0.41 (62.05$\times$) & \CAM 1.163 & \CAM 0.69 (50.12$\times$) & \CAM 1.034 \\ & \CAM CAM-30 & \CAM 1.31 (24.65$\times$) & \CAM 1.023 & \CAM 0.19 (104.63$\times$) & \CAM 1.749 & \CAM 0.90 (28.27$\times$) & \CAM 1.024 & \CAM 1.72 (20.10$\times$) & \CAM 1.018 \\ & \CAM CAM-50 & \CAM 2.11 (15.30$\times$) & \CAM 1.030 & \CAM 0.28 (71.00$\times$) & \CAM 1.359 & \CAM 1.47 (17.31$\times$) & \CAM 1.012 & \CAM 2.80 (12.35$\times$) & \CAM 1.027 \\ & \CAM CAM-100 & \CAM 4.11 (7.86$\times$) & \CAM 1.035 & \CAM 0.63 (31.56$\times$) & \CAM 1.077 & \CAM 2.87 (8.86$\times$) & \CAM 1.023 & \CAM 5.50 (6.29$\times$) & \CAM 1.033 \\ & \REP Replay-10 & \REP 14.24 (2.27$\times$) & \REP 1.001 & \REP 12.66 (1.57$\times$) & \REP 1.313 & \REP 17.34 (1.47$\times$) & \REP 1.005 & \REP 13.71 (2.52$\times$) & \REP 1.001 \\ & \REP Replay-30 & \REP 18.21 (1.77$\times$) & \REP 1.000 & \REP 14.27 (1.39$\times$) & \REP 1.108 & \REP 15.95 (1.59$\times$) & \REP 1.001 & \REP 17.68 (1.96$\times$) & \REP 1.001 \\ & \REP Replay-50 & \REP 22.27 (1.45$\times$) & \REP 1.000 & \REP 15.59 (1.28$\times$) & \REP 1.049 & \REP 18.36 (1.39$\times$) & \REP 1.001 & \REP 20.92 (1.65$\times$) & \REP 1.000 \\ & \REP Replay-100 & \REP 32.29 (1.00$\times$) & \REP 1.000 & \REP 19.88 (1.00$\times$) & \REP 1.000 & \REP 25.44 (1.00$\times$) & \REP 1.000 & \REP 34.58 (1.00$\times$) & \REP 1.000 \\ & \LPM LPM & \LPM 22.74 (1.42$\times$) & \LPM 1.027 & \LPM 17.46 (1.14$\times$) & \LPM 20.657 & \LPM 20.40 (1.25$\times$) & \LPM 2.537 & \LPM 21.58 (1.60$\times$) & \LPM 1.204 \\ \midrule \multirow{9}{*}{wiki} & \CAM CAM-10 & \CAM 0.55 (45.93$\times$) & \CAM 1.021 & \CAM 0.10 (143.80$\times$) & \CAM 4.052 & \CAM 0.42 (45.86$\times$) & \CAM 1.147 & \CAM 0.69 (32.32$\times$) & \CAM 1.031 \\ & \CAM CAM-30 & \CAM 1.31 (19.28$\times$) & \CAM 1.023 & \CAM 0.19 (75.68$\times$) & \CAM 1.734 & \CAM 0.92 (20.93$\times$) & \CAM 1.022 & \CAM 1.72 (12.97$\times$) & \CAM 1.020 \\ & \CAM CAM-50 & \CAM 2.11 (11.97$\times$) & \CAM 1.030 & \CAM 0.28 (51.36$\times$) & \CAM 1.362 & \CAM 1.58 (12.19$\times$) & \CAM 1.013 & \CAM 2.83 (7.88$\times$) & \CAM 1.028 \\ & \CAM CAM-100 & \CAM 4.20 (6.01$\times$) & \CAM 1.035 & \CAM 0.62 (23.19$\times$) & \CAM 1.077 & \CAM 3.01 (6.40$\times$) & \CAM 1.024 & \CAM 5.56 (4.01$\times$) & \CAM 1.033 \\ & \REP Replay-10 & \REP 9.88 (2.56$\times$) & \REP 1.002 & \REP 8.30 (1.73$\times$) & \REP 1.318 & \REP 8.84 (2.18$\times$) & \REP 1.009 & \REP 9.11 (2.45$\times$) & \REP 1.003 \\ & \REP Replay-30 & \REP 13.59 (1.86$\times$) & \REP 1.000 & \REP 9.74 (1.48$\times$) & \REP 1.117 & \REP 11.07 (1.74$\times$) & \REP 1.003 & \REP 12.30 (1.81$\times$) & \REP 1.001 \\ & \REP Replay-50 & \REP 22.50 (1.12$\times$) & \REP 1.000 & \REP 10.89 (1.32$\times$) & \REP 1.047 & \REP 13.87 (1.39$\times$) & \REP 1.002 & \REP 16.10 (1.39$\times$) & \REP 1.001 \\ & \REP Replay-100 & \REP 25.26 (1.00$\times$) & \REP 1.000 & \REP 14.38 (1.00$\times$) & \REP 1.000 & \REP 19.26 (1.00$\times$) & \REP 1.000 & \REP 22.30 (1.00$\times$) & \REP 1.000 \\ & \LPM LPM & \LPM 14.70 (1.72$\times$) & \LPM 1.027 & \LPM 12.23 (1.18$\times$) & \LPM 20.661 & \LPM 14.01 (1.37$\times$) & \LPM 2.595 & \LPM 14.28 (1.56$\times$) & \LPM 1.194 \\ \bottomrule \end{tabular} } 
\vspace{-2ex}
\end{table*}

\subsection{Evaluation of Tuning Results} \label{subsec:tuning-results}

\begin{figure}[t]
    \centering
    \includegraphics[width=0.8\linewidth]{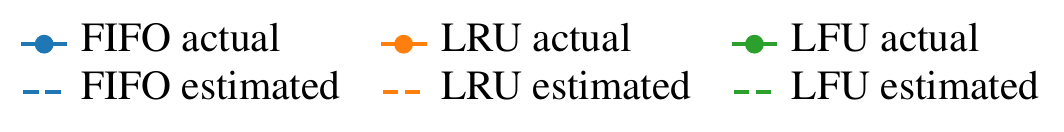}
    \includegraphics[width=0.8\linewidth]{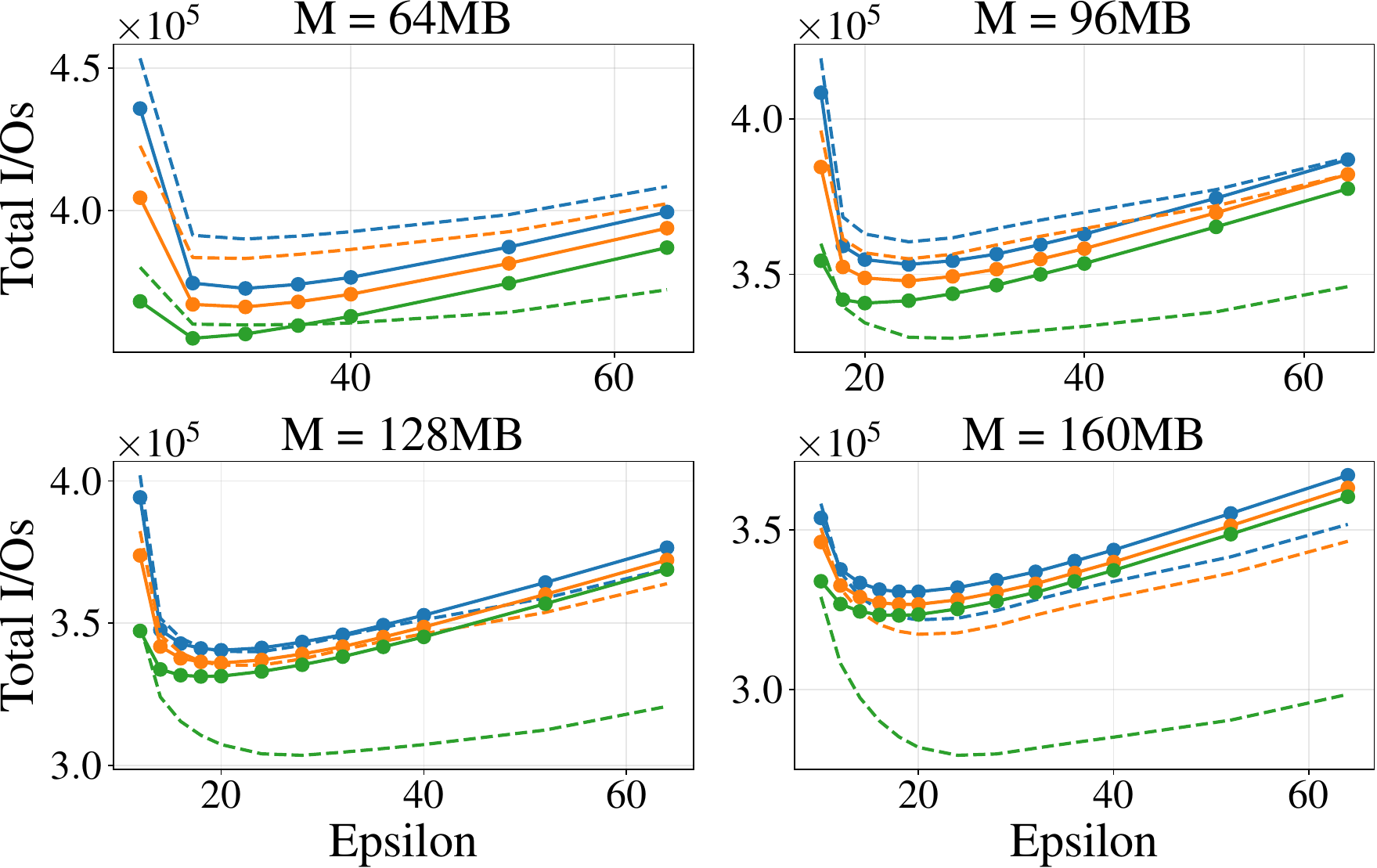}
    \caption{Estimated (CAM-100) and actual I/O costs (PGM) across buffer sizes and eviction policies (books, w4).}
    \label{fig:logical_ios_vs_epsilon}
\end{figure}

\begin{figure}[t]
    \centering
    \includegraphics[width=0.8\linewidth]{figures/logical_ios_vs_epsilon_legend.pdf}
    \includegraphics[width=0.8\linewidth]{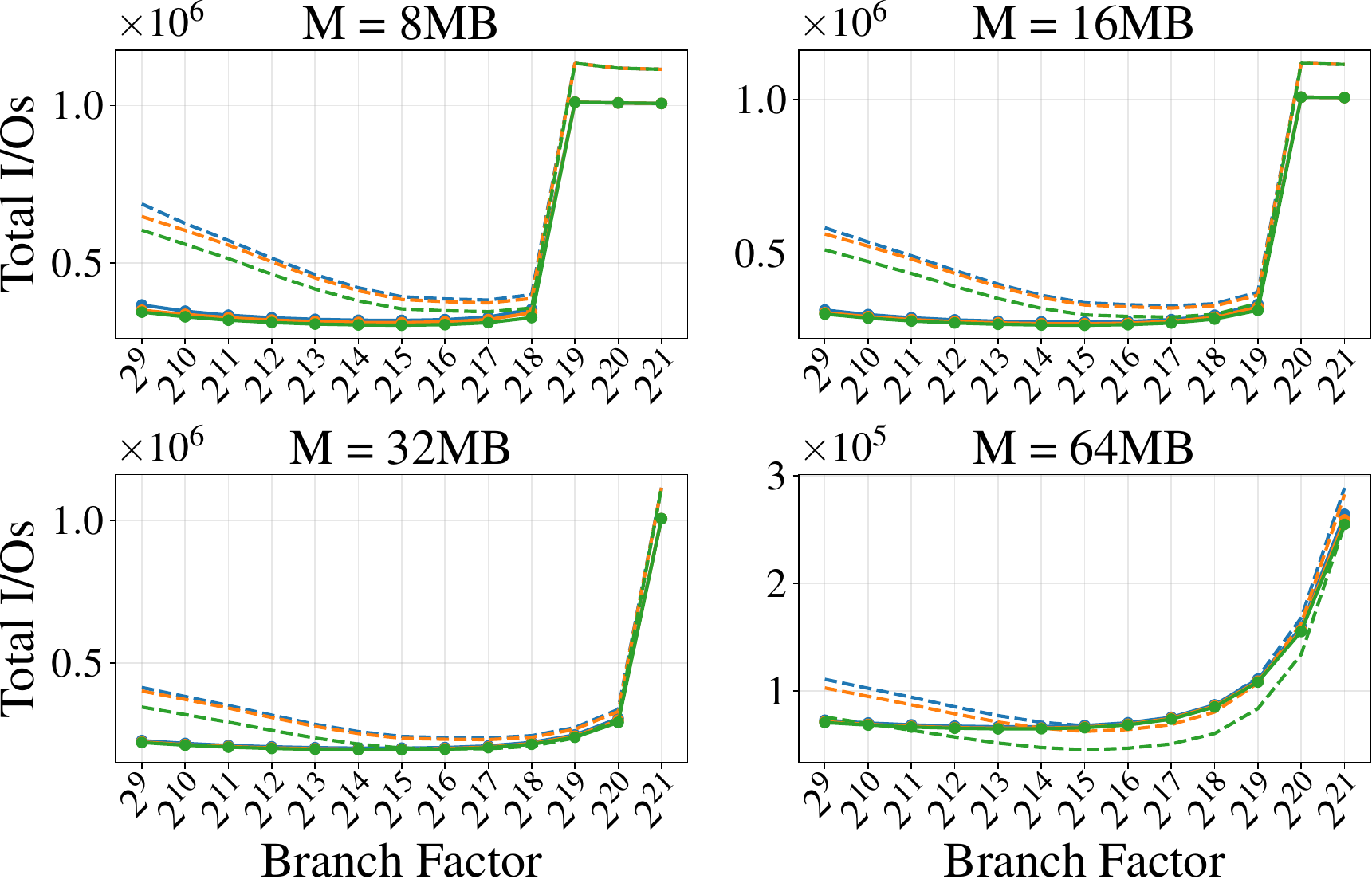}
    \caption{Estimated (CAM-100) and actual I/O costs (RMI) across buffer sizes and eviction policies (books, w4). The sharp increase at large branch factors occurs when the RMI size leaves little or no space for the page buffer.}
    \label{fig:RMI-tuning}
\end{figure}

\noindent
\textit{\textbf{Validation of Methodology.}}
\Cref{fig:logical_ios_vs_epsilon} evaluates CAM-guided PGM tuning on the books dataset under workload w4 across memory budgets from 64\,MB to 160\,MB. 
CAM's estimated I/O costs closely track the actual measured totals across FIFO and LRU, capturing the characteristic U-shaped trade-off in $\varepsilon$. LFU exhibits larger estimation errors because its analytical model assumes a converged steady state, while the measured workloads are finite and may not fully realize the long-term LFU behavior. Despite this discrepancy, the estimated and measured curves follow the same overall trend and lead to similar tuning decisions.

Initially, increasing $\varepsilon$ shrinks the index footprint and enlarges the buffer, thereby improving hit rates. Beyond a knee point, however, the growing last-mile search window dominates and total I/O rises. This alignment across all buffer sizes and eviction policies confirms that CAM reliably identifies the optimal error bound for PGM. 
Similarly, for RMI tuning, we choose the branching factor as the tuning knob of a two-layer RMI with linear-spline leaf models. 
Similar to the PGM case, the I/O cost estimated by CAM 
(as discussed in \cref{subsec:rmi_tuning}) closely tracks actual I/O across 
configurations (\cref{fig:RMI-tuning}), validating CAM-based tuning for I/O-aware settings. 

\noindent
\textit{\textbf{Evaluation Details.}}
To evaluate the practical impact of CAM-guided tuning, we compare it against the multicriteria PGM optimizer~\cite{DBLP:journals/pvldb/FerraginaV20} for PGM-index and CDFShop~\cite{DBLP:conf/sigmod/MarcusZK20} for RMI.
These baselines optimize for index size and lookup cost but do not explicitly model the interaction between index footprint and page buffer capacity under a fixed memory budget $M$. 
For a fair comparison, we reserve a fixed fraction of $M$ as the page buffer and pass the remainder to the baseline tuner as its index-space constraint. 
This ensures that baseline methods operate in the same execution environment where buffer and index share the available memory. 
In contrast, CAM evaluates each candidate configuration by estimating its effective I/O cost after accounting for both the index footprint and the resulting buffer capacity. 
We report the lookup query throughput of indexes optimized by different tuners, and the results are given in \cref{fig:tuning} and \cref{fig:rmi-tuning}.

\noindent
\textit{\textbf{PGM Tuning Results.}}
As shown in \cref{fig:tuning}, CAM achieves the best overall throughput across all tested memory budgets on the books dataset under workload w4. 
Against the multicriteria PGM optimizer~\cite{DBLP:journals/pvldb/FerraginaV20}, CAM improves QPS by up to \textbf{1.17$\times$} and 
reduces tuning time by \textbf{75.7\%}.
The improvement comes from CAM's ability to choose configurations that better balance last-mile search cost and cache effectiveness. 
In contrast, the multicriteria PGM tuner optimizes $\varepsilon$ mainly from the perspective of index size and last-mile lookup cost under a fixed memory split, ignoring the interaction between last-mile search and cache policy. 
As a result, the selected configuration is suboptimal for the complete disk-resident query pipeline.
% Interestingly, CAM may select a slightly larger $\varepsilon$ than the baseline PGM tuner. 
% Although this increases the last-mile search range, it also reduces the number of PGM segments and leaves more memory for caching data pages. 
% When the additional logical page references caused by a larger $\varepsilon$ are mostly absorbed by the buffer, the reduced index footprint leads to fewer physical I/Os and higher throughput. 
% By contrast, the multicriteria PGM tuner optimizes $\varepsilon$ mainly from the perspective of index size and last-mile lookup cost, and therefore may select a configuration that is locally optimal for the index but suboptimal for the complete disk-resident query pipeline.

% a slightly larger $\varepsilon$ may increase logical page references, but it can also reduce the index footprint and leave more memory for buffering data pages, thereby reducing physical I/Os and improving throughput.

\begin{figure}[t]
    \centering
    \includegraphics[width=\linewidth]{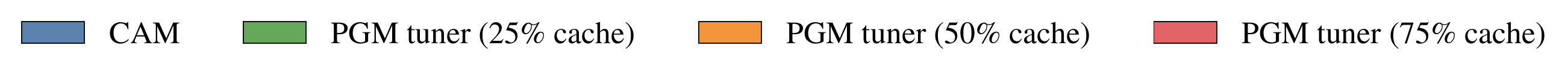}

    \begin{subfigure}{0.48\linewidth}
        \centering
        \includegraphics[width=\linewidth]{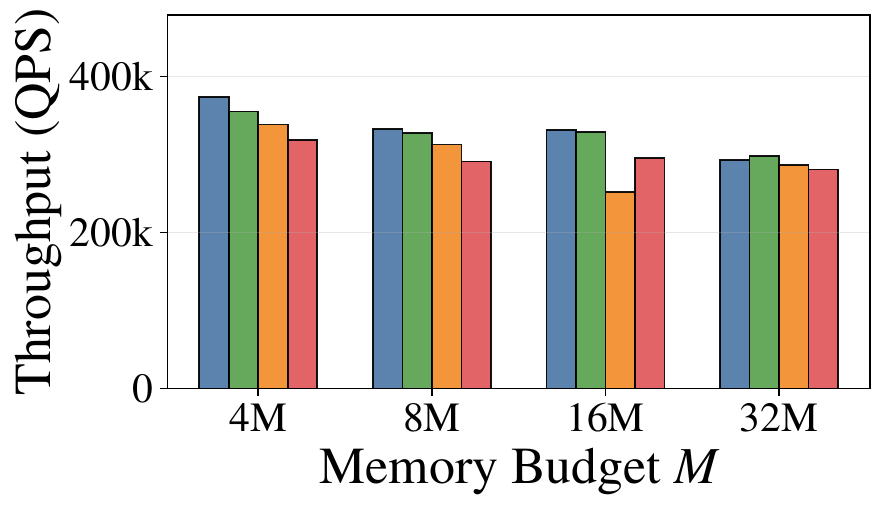}
        \caption{Throughput comparison.}
        \label{fig:tuning-throughput}
    \end{subfigure}
    \hfill
    \begin{subfigure}{0.48\linewidth}
        \centering
        \includegraphics[width=\linewidth]{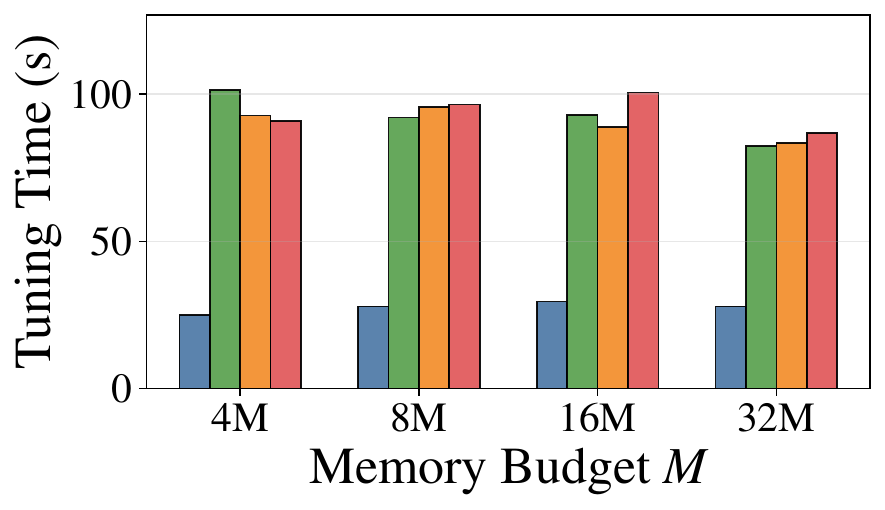}
        \caption{Tuning time comparison.}
        \label{fig:tuning-time}
    \end{subfigure}
    \caption{Comparison between CAM-guided tuning and multicriteria PGM tuning under different memory budgets on the books dataset 
    workload w4.}
    \label{fig:tuning}
    \vspace{-2ex}
\end{figure}

\noindent
\textit{\textbf{RMI Tuning Results.}}
\Cref{fig:rmi-tuning} compares CAM against CDFShop~\cite{DBLP:conf/sigmod/MarcusZK20} under varying memory budgets. 
CAM outperforms all CDFShop variants across evaluated budgets, achieving up to \textbf{1.66$\times$} higher throughput. It also reduces tuning time by \textbf{60.1\%} compared with CDFShop (\cref{fig:rmi-tuning}(b)). 
This reduction is smaller than in the PGM tuning case because RMI-based CAM (\cref{subsec:rmi_tuning}) requires physically constructing each candidate index to determine its error bounds, whereas PGM tuning only requires a small number of sampled constructions to fit the index-size model.

The key advantage of CAM here is similar.
CDFShop performs RMI structure search to find efficient designs that balance lookup latency, model size, and prediction error, while ignoring practical I/O behaviors.
This limitation is amplified for RMI because the tuning knobs like branch factor is typically selected from a sparse set of power-of-two candidates.
As a result, a small change in the index-space constraint may switch the selected branch factor to a substantially different configuration, and a suboptimal branch-factor choice can lead to a large performance gap.
% This results indicate that RMI tuning is particularly sensitive to the memory allocation between model storage and the page buffer. 
% A configuration with lower prediction error (i.e., more accurate CDF approximation) does not necessarily minimize end-to-end query cost once the data fully resides on disk. 
% If the selected RMI configuration consumes too much memory, the remaining buffer becomes too small to retain frequently accessed data pages, increasing physical I/Os. 
% CAM avoids this by evaluating each candidate configuration according to its effective I/O cost, rather than its index-side lookup cost alone.

% \begin{figure}[t]
%     \centering
%     \includegraphics[width=0.45\textwidth]{figures/exp/rmi_tuning_cmp_legend.pdf}
%     \includegraphics[width=0.45\textwidth]{figures/exp/rmi_tuning_cmp.pdf}
%     \caption{Throughput comparison between CAM-guided tuning and CDFShop tuning for RMI under different memory budgets on the books dataset with workload w4.}
%     \label{fig:rmi-tuning}
% \end{figure}

\begin{figure}[t]
    \centering
    \includegraphics[width=\linewidth]{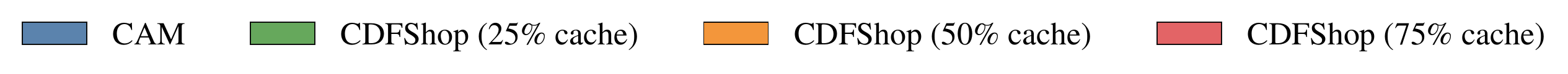}

    \begin{subfigure}{0.48\linewidth}
        \centering
        \includegraphics[width=\linewidth]{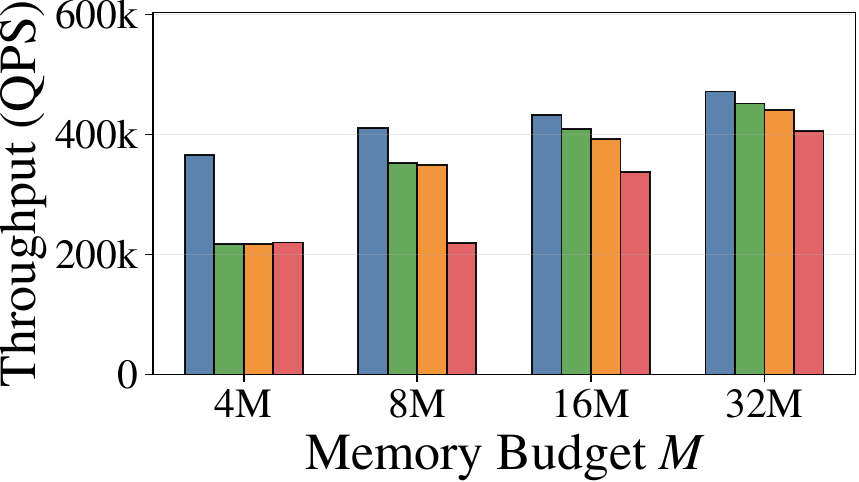}
        \caption{Throughput comparison.}
        \label{fig:rmi-tuning-throughput}
    \end{subfigure}
    \hfill
    \begin{subfigure}{0.48\linewidth}
        \centering
        \includegraphics[width=\linewidth]{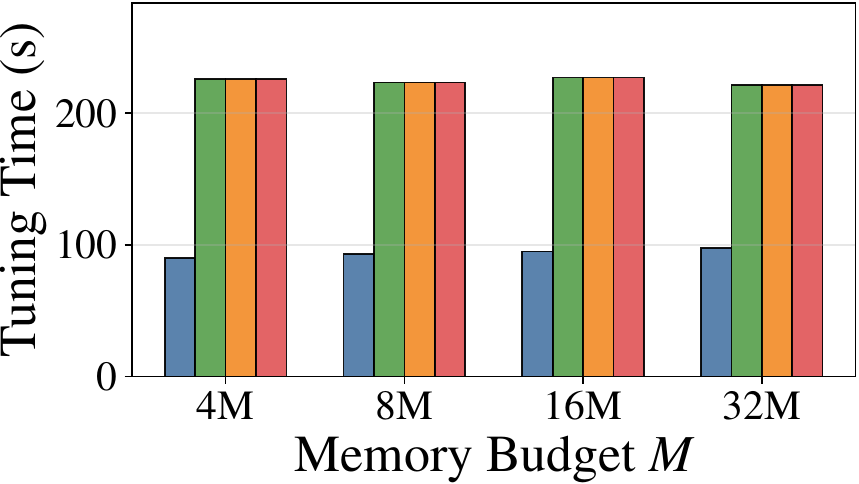}
        \caption{Tuning time comparison.}
        \label{fig:rmi-tuning-time}
    \end{subfigure}

    \caption{Comparison between CAM-guided tuning and CDFShop tuning for RMI under different memory budgets on the books dataset with workload w4. Notably, CDFShop's tuning time only depends on the specific input dataset.}
    \label{fig:rmi-tuning}
\end{figure}

% Overall, existing cache-oblivious tuners~\cite{DBLP:journals/pvldb/FerraginaV20,DBLP:conf/sigmod/MarcusZK20} sacrifice measurable performance even when granted an explicit buffer allocation. 
% CAM closes this gap by extending the tuning objective from index-local metrics to end-to-end cache-aware I/O cost. 

% The benefit is modest but consistent for PGM, whose primary knob is the error bound $\varepsilon$, 
% and more pronounced for RMI, whose configuration space couples more tightly to memory footprint and buffer efficiency.

\subsection{Evaluation of Join Processing}\label{sec:exp-hybrid-join}
\noindent
\textit{\textbf{Evaluation Details.}}
We evaluate the hybrid join on the books dataset under six workload mixtures (w1-w6; \cref{tab:workload_mixture}). 
% Results on other datasets are similar and thus omitted.
We compare four execution strategies. 
\textsc{INLJ} is the conventional index nested-loop join that probes the learned index using the original outer-key order.
\textsc{Point-Only} first sorts the outer keys and then issues one indexed point lookup per key; it can be viewed as a sorted INLJ baseline.
\textsc{Range-Only} sorts the outer keys and replaces them by a range probe between its two endpoints, followed by filtering; it resembles a sort-merge-style strategy.
\textsc{Hybrid} is our proposed strategy, which adaptively chooses point or range probing for each partition according to the estimated cost.

Following \cref{sec:join-query}, we partition sorted probe keys into segments with a minimum size of 1024 keys to prevent fragmentation. 
The cost-model parameters in \cref{eq:cost_for_join} are fit once on a small calibration set and reused across all workloads. 
Per-page miss latencies $\lambda_{\text{point}}$ and $\lambda_{\text{range}}$ are estimated as the median ratio of I/O time to physical I/O count across calibration runs. 
We then subtract the fitted I/O component from end-to-end time to isolate CPU cost and fit the remaining coefficients via ordinary least 
squares. 
The fitted parameters are summarized in \cref{tab:workload_mixture}.

\noindent
\textit{\textbf{Evaluation Results.}}
\cref{fig:join} shows that sorted probing (both point-only and range-only) consistently outperforms a conventional INLJ method (without sorting) across all workloads, because sorting the outer keys improves locality and substantially increases buffer reuse and the sorting overhead is marginal compared with these gains (less than 10\%). 
Across all workloads, the proposed Hybrid strategy (\cref{sec:join-query}) achieves the lowest end-to-end join time, outperforming \textsc{INLJ} by up to \textbf{8.8$\times$}. The gains come from cost-guided segmentation: the strategy adapts to local key density, using point probes to avoid overfetch in sparse regions and range probes to amortize overhead in dense regions. 
% We further observe that, across all workloads, the Range-Only strategy consistently underperforms Point-Only probing, indicating that applying range probing globally results in substantial redundant I/O. In contrast, the hybrid strategy restricts range probing to carefully selected dense segments, effectively avoiding unnecessary page fetches.
\begin{figure}[t]
    \centering
    \includegraphics[width=0.4\textwidth]{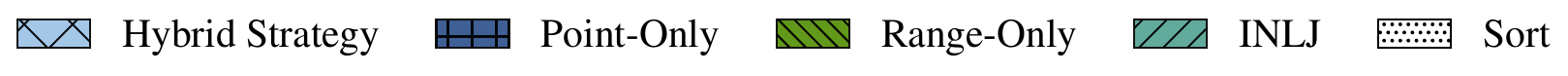}
    \includegraphics[width=0.28\textwidth]{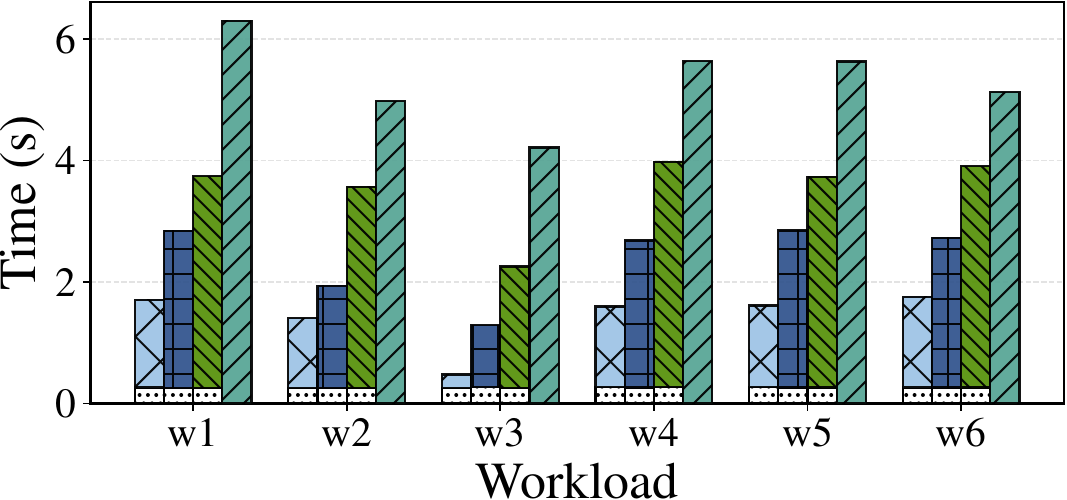}
    \caption{End-to-end evaluation for joining a 1M outer relation (generated using workloads w1-w6) against a 200M indexed inner relation using a 16 MB page buffer.}
    \label{fig:join}
\end{figure}

% \noindent
% \textit{\textbf{Impact of Workload Patterns.}}
% The relative effectiveness of different join strategies varies noticeably across different workloads.
% Under highly skewed workloads (w2 and w3), the advantage of hybrid execution is particularly pronounced.
% Such workloads form a small number of dense key regions, which are accurately identified by the hybrid strategy and efficiently processed using range probes. 
% In more uniform or mixed workloads (w1, w4, w5 and w6), dense regions are less prominent; consequently, the hybrid strategy behaves similarly to point-only and range-only probing, yielding smaller yet still consistent performance improvements.

\section{Conclusion}    \label{sec:conclusion}
In this paper, we proposed \textbf{CAM}, a cache-aware I/O cost model that combines last-mile page-access analysis with buffer hit-rate estimation. CAM provides a lightweight way to estimate the effective physical I/O cost of learned-index queries without replaying the full workload. For PGM-index, CAM-guided tuning improves throughput by up to \textbf{1.17$\times$} over the multicriteria PGM tuning. For RMI, CAM improves throughput by up to \textbf{1.66$\times$} over CDFShop with I/O-related considerations. Finally, we proposed a hybrid join strategy that adaptively chooses point or range probes according to local key density, improving end-to-end join performance by up to \textbf{8.8$\times$} over unsorted INLJ. Overall, CAM provides a principled foundation for memory-aware learned-index tuning and query optimization in external-memory database systems.

% \section*{Acknowledgment}

% The preferred spelling of the word ``acknowledgment'' in America is without 
% an ``e'' after the ``g''. Avoid the stilted expression ``one of us (R. B. 
% G.) thanks $\ldots$''. Instead, try ``R. B. G. thanks$\ldots$''. Put sponsor 
% acknowledgments in the unnumbered footnote on the first page.

\bibliographystyle{IEEEtran}
\bibliography{ref}

\end{document}